\documentclass[10pt,pra,aps,groupaddress,twocolumn,letter,nopacs,nofootinbib,longbibliography,notitlepage]{revtex4-2}

\usepackage{graphicx} 
\usepackage{booktabs}
\usepackage{multirow}
\usepackage{array}
\usepackage{makecell}
\usepackage{amsmath}
\usepackage{amssymb}
\usepackage{mathtools}
\usepackage{hyperref}
\usepackage[capitalise]{cleveref}
\usepackage{amsthm}
\usepackage{braket}
\usepackage{tikz}
\usepackage[ruled,vlined]{algorithm2e}
\usepackage{tabularx}
\usepackage{float}

\newtheorem{prop}{Proposition}
\newcommand{\appref}[1]{Appendix~\ref{#1}}

\begin{document}

\title{Vine Codes: Low-Overhead Quantum LDPC Codes on a Planar Square Grid}

\author{Georgia M. Nixon}
\email{georgia.nixon@sydney.edu.au}

\author{Campbell K. McLauchlan}
\email{campbell.mclauchlan@sydney.edu.au}

\author{Charles C. L. van Rest}
\affiliation{School of Physics, The University of Sydney, Sydney, New South Wales 2006, Australia}

\begin{abstract}
The surface code is a promising route towards large-scale quantum computing, requiring only nearest-neighbour gates amenable to superconducting hardware. However, surface codes incur large  qubit overheads. Novel quantum low-density parity check (qLDPC) codes promise to reduce overheads, but require long-range connections that are difficult to achieve on superconducting platforms.
Here, we introduce  ``Vine Codes" -- qLDPC codes that are implementable on a planar square grid through nearest-neighbour, two-qubit gates native to superconducting platforms (iSWAP and CZ). Our approach generalises the recently introduced ``Directional Codes"~\cite{Geher2025Directional} which are constrained to a torus. In contrast, vine codes have open boundary conditions constructed with the aid of routing qubits. 
We perform extensive numeric searches and find promising candidate vine codes, e.g. $[[121,4,6]]$, $[[221,6,7]]$, and $[[234,9,6]]$ codes. 
We verify the circuit distances and show that data and measure qubits required can be reduced by up to $\sim28\%$ relative to the surface code at a circuit distance of $7$. Even including routing qubits, vine codes require fewer total qubits than the surface code (e.g. $\sim18\%$ reduction at circuit distance 10) and benefits are expected to increase at higher distances. 
We perform circuit-level noise simulations to demonstrate that under a realistic noise model and at a near-term noise rate of $10^{-3}$, vine codes can perform better than the surface code while using fewer qubits. 
We give an exhaustive list of all unique vine codes up to stabiliser-weight $9$.
We additionally introduce ``Flip-Vine Codes" which possess single-qubit transversal Clifford gates useful for fault-tolerant logic and magic state cultivation. 
We furthermore construct examples of generalised open boundaries for vine codes that go beyond the familiar $X$/$Z$ boundaries of the surface and tile codes.
\end{abstract}

\maketitle
\section{Introduction}

\begin{figure*}[t]
    \centering
    \includegraphics{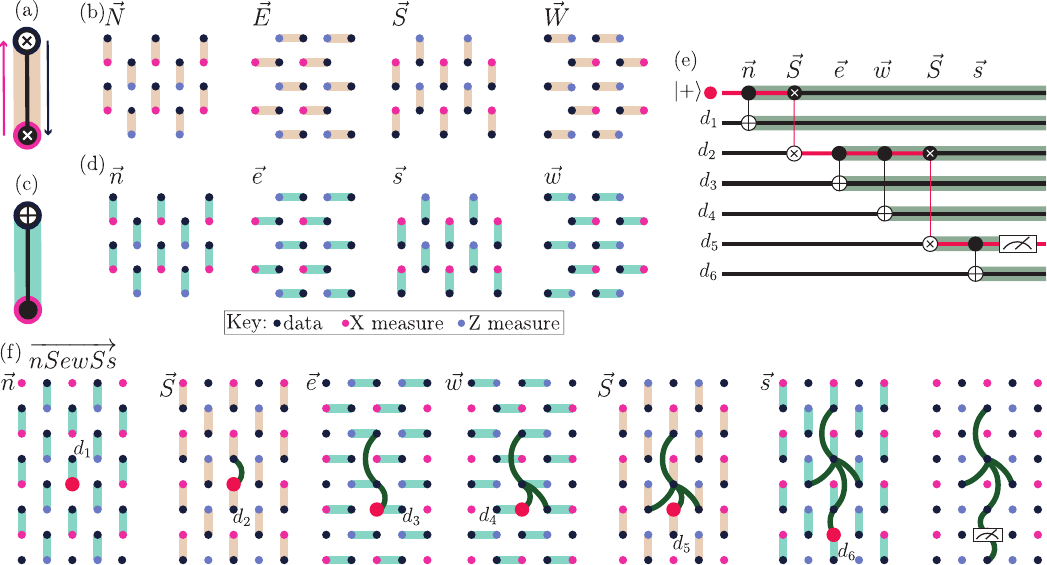}
    \caption{
    Vine codes are defined with qubits on a square lattice where each measure qubit has four nearest neighbour data qubits. 
    \textbf{(a)}-\textbf{(b)} The uppercase steps $\vec{N}$, $\vec{E}$, $\vec{S}$, $\vec{W}$ entangle pairs of data and measure qubits via a CXSWAP gate (see (a)) following the data-measure qubit pairings shown in (b). The labels $\vec{N}$, $\vec{E}$, $\vec{S}$ and $\vec{W}$ refer to cardinal directions that mirror the vector from the measure to the data qubit in each case. A pink $X$-type measure qubit will control the CXSWAP as in (a), while a blue $Z$-type measure qubit will be the target of the CXSWAP. 
    \textbf{(c)}-\textbf{(d)} The lowercase steps $\vec{n}$, $\vec{e}$, $\vec{s}$, $\vec{w}$ entangle pairs of data and measure qubits via a CX gate (see (c)), where now no qubits move. The data-measure qubit pairings for each instruction are shown in (d), following the same pattern as (c). 
    \textbf{(e)} The qubit register showing the dynamic position of the red $X$-type measure qubit for the step sequence $\overrightarrow{nSewSs}$. Part of the $Z$-error detecting region for this stabiliser is highlighted in green.
    \textbf{(f)} $\overrightarrow{nSewSs}$ is a step sequence that encodes a valid vine code on the layout shown (alternating rows of $X$- and $Z$-type measure qubits). We show the movement of an $X$ measure qubit in red, and label the data qubits it entangles with $d_i$ in each step $i$. The support of the stabiliser is shown by the data qubits connected by the green vine as indicated.  
     } 
    \label{fig:vine_codes}
\end{figure*}

To run useful, large-scale quantum algorithms, quantum error-correction (QEC) is required to suppress the noise encountered in physical devices. A leading QEC code candidate is the surface code~\cite{Kitaev2003Fault,Dennis2002Topological,Bravyi1998Quantum,Fowler2012Surface,Campbell2017Roads}, owing to its high threshold, 2D-local connectivity requirements, amenability to fast and accurate decoders~\cite{Dennis2002Topological,PyMatching2023,Delfosse_2021_union_find,shutty2024efficientnearoptimaldecodingsurface,Bravyi_2014_MLD_SC}, and the existence of mature fault-tolerant computation schemes~\cite{Raussendorf2007Fault,Horsman2012Lattice,Bombin_2009,Brown_2017,Brown2020FaultTolerant,Bravyi_2005,Litinski_2019_MSD,Litinski_2019_GoSC,Chamberland_2022_twist_free,Geher_tangling_schedules_2024,low2026denserplanarsurfacecode}. The 2D planar layout of the surface code is particularly well-suited to superconducting platforms, where high-degree or non-local connectivity between qubits remains a challenge. Superconducting devices benefit from fast and high-fidelity gates, and several core primitives of fault-tolerant quantum computation have already been demonstrated on these devices~\cite{Google2023BelowThreshold,Gupta_2024,Google_SC_2024,eickbusch2025demonstratingdynamicsurfacecodes,rosenfeld2025magicstatecultivationsuperconducting,Harper_2026_Characterising,lee2026scalablequantumerrorcorrection,wang2026LatticeSurgery,Wang_2026_BB_demo}. They are therefore a promising platform for future quantum computers at the utility scale.  

Recent advancements in quantum low-density parity check (qLDPC) codes have demonstrated  high rate, high distance codes at intermediate scales that can dramatically reduce the large overheads of QEC associated with the surface code~\cite{Bravyi2024HighThreshold,webster2026pinnaclearchitecturereducingcost,yoder2025tourgrossmodularquantum,xu2026distillingmagicstatesbicycle,cain2026shorsalgorithmpossible10000,tan2025singleshotuniversalityquantumldpc,Menon_2026,jacob2026singleshotdecodingfaulttolerantgates,malcolm2025computingefficientlyqldpccodes,Scruby_2026,Geher2025Directional,steffan2025tilecodes,liang2025planarquantumlowdensityparitycheck,radebold2025explicitinstancesquantumtanner,Panteleev2021Degenerate,Tremblay2022Constant,Pecorari2025HighRate,2BGA_codes}. These codes may help reduce the overheads of running useful quantum algorithms by orders of magnitude~\cite{webster2026pinnaclearchitecturereducingcost,yoder2025tourgrossmodularquantum,cain2026shorsalgorithmpossible10000}. However, this comes with the tradeoff that many highly performant qLDPC codes, such as the celebrated bivariate bicycle codes~\cite{Bravyi2024HighThreshold}, require non-local connectivity when the qubits are embedded in two or even three dimensional space. This poses a particular challenge for superconducting architectures which have fixed qubit layouts built into a chip. While there exist demonstrations of bivariate bicycle codes in superconducting platforms~\cite{Wang_2026_BB_demo}, the logical fidelities demonstrated are far lower than is possible with the surface code~\cite{Google_SC_2024}. In general, the technology required for high-fidelity, non-local connections between superconducting qubits remains a challenge to realise.


Directional codes are a family of recently-introduced qLDPC codes defined on a torus that can achieve high-rate, high-distance codes while only requiring local connectivity~\cite{Geher2025Directional,rowshan2026directional}. 
They achieve this by using iSWAP gates for syndrome measurement. The iSWAP, along with the CZ gate, is native to superconducting platforms, and is a strong contender for the backbone of syndrome extraction circuits~\cite{McEwen_2023_relaxing,eickbusch2025demonstratingdynamicsurfacecodes,iswap_1,iswap_2,iswap_3,iswap_with_cz}. 
By re-casting the qubit-swap mechanism in the iSWAP gate as an opportunity rather than a hindrance,  directional codes can achieve stabiliser support spanning further than nearest neighbours as measure and data qubit positions are dynamically updated throughout the syndrome extraction cycle.
However, directional codes still require non-local connectivity on a planar chip as they are embedded on a torus. While techniques exist for enforcing open boundary conditions on 2D translationally invariant codes originally defined on toric layouts~\cite{steffan2025tilecodes,liang2025planarquantumlowdensityparitycheck,eberhardt2024pruningqldpccodesbivariate}, it is not clear whether directional codes will be amenable to these techniques while retaining 2D locality and low-depth syndrome measurement circuits. Moreover, since encoding rate and distances, along with circuit-level performance, can differ dramatically between open and closed boundary conditions, it is unclear whether the gains achieved by directional codes will transfer to the cases of practical relevance for superconducting devices. 

We address this challenge in this manuscript by introducing vine codes (so-named because of the shape of their stabiliser supports; see Fig.~\ref{fig:vine_codes}). Vine codes demonstrate that it is possible to achieve higher rate and distance codes than the surface code without requiring any other connectivity than a 2D, planar square grid. These codes generalise the directional codes in several ways. Firstly, vine codes are defined with open boundaries constructed by introducing routing qubits. Secondly, by utilising both iSWAP and CZ gates for syndrome extraction, vine codes form a larger family of qLDPC codes. 

We extend the utility of vine constructions by also introducing ``flip-vine codes". By altering the roles of measure qubits (between measuring $X$ and $Z$ stabilisers) in each QEC round, flip-vine codes are weakly self-dual and possess single-qubit transversal Clifford gates. These features are useful for fault-tolerant logic~\cite{Kubica_2015} and for magic state cultivation~\cite{gidney2024magicstatecultivationgrowing,claes2025cultivatingtstatessurface,vaknin2026efficientmagicstatecultivation,sahay2026foldtransversalsurfacecodecultivation}. 
We provide examples of flip-vine codes with periodic boundary conditions that support transversal single-qubit Clifford gates. These reduce data qubit numbers by up to $\sim 83.3\%$ relative to the surface code.

Through numerical searches, we find high-performing vine codes with lower overheads than the surface code (see Figs.~\ref{fig:nSeWwSs_results_fig2} and~\ref{fig:nSeWWwSs_results_fig3}).  We construct vine codes with  parameters such as $[[221,6,7]]$, $[[121, 4, 6]]$, $[[234,9,6]]$ and $[[437, 6, 10]]$. By providing exact circuit distances for small-instances at distance $7$, we find vine codes that reduce the data and measure qubits required by up to $28\%$ relative to the surface code.  We demonstrate that these efficiencies over the surface code are expected to improve further for larger distances. We find examples of vine codes on boundary geometries with no distance-reducing hook errors. 

Vine codes utilise routing qubits to achieve an open boundary.
Including routing qubits, vine codes still maintain a lower total overhead relative to the surface code. Remarkably, we find vine code examples that require $\sim 18\%$ fewer total qubits (including routing) than the surface code at distance 10 (in a code which we expect to have no distance-reducing hooks). Again, these benefits are expected to improve for larger distances.
We introduce various techniques that help minimise the number of routing qubits required for syndrome measurement, reducing it by up to $\sim66\%$ in relevant examples, compared to naive requirements.

We perform circuit-based simulations of several examples of vine codes to demonstrate 
their expected performance under a superconducting-inspired noise model. At a noise rate of $p = 10^{-3}$, accessible to near-term devices~\cite{ibm_compute},  vine codes perform better or comparably to the surface code despite using fewer qubits. 

We provide detailed descriptions of how to construct open boundaries for vine codes. We distinguish between ``Pauli boundaries,"  analogous to the $X$/$Z$ boundaries of the surface code, and ``generalised boundaries," which encompass the coloured boundaries of the 2D color code~\cite{2DCC_bombin_delgado,Kesselring_2024_anyon_condensation,Kesselring_2018_bdrys_twists}. We provide examples of these generalised boundaries for (flip-)vine codes and show that the stabilisers in these examples can be measured with no increase in the circuit depth.

The remainder of this paper is structured as follows. 
In \cref{sec:vinecodes}, we introduce vine codes (on an infinite plane) and flip-vine codes. In \cref{sec:boundaryconstruction}, we describe both the Pauli and generalised boundary constructions. In \cref{sec:results}, we report the results of an exhaustive search of all valid vine codes up to stabiliser weight $9$, and report example vine and flip-vine codes that out-perform the surface code in terms of qubit overhead. We also report the results of our numerical simulations. 
In \cref{sec:conclusion}, we conclude and discuss future directions. In the appendices, we provide further details of our numeric searches for valid vine codes, the geometric properties of vine codes, the boundary construction at corners, procedures for routing qubit minimisation, and the generalised boundary circuits.

\begin{figure*}[ht]
    \centering
    \includegraphics{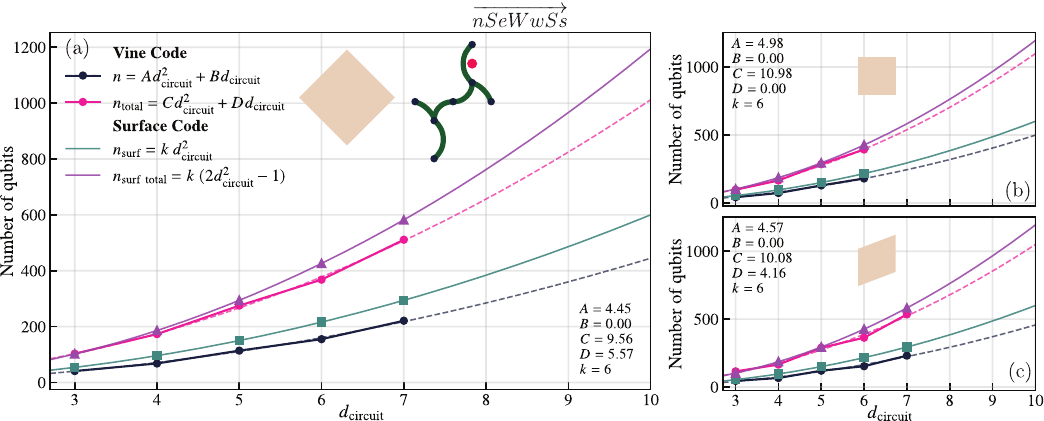}
    \caption{The qubit overhead scaling of the $\overrightarrow{nSeWwSs}$ vine code for three different open boundary geometries, each encoding $k=6$ logical qubits. The stabiliser support is indicated schematically by the green vine and red measure qubit in (a). The brown shape at the top of each subfigure indicates the open boundary geometry. For each boundary geometry (see \cref{fig:code_patches}), we plot the total number of qubits $n_{\mathrm{total}}$ (including data, measure and routing qubits) and number of data qubits $n$ needed to create a patch of circuit distance $d_\mathrm{circuit}$, comparing these to the total number of qubits and data qubits, respectively, needed for $k$ surface codes with the same distance. We fit the vine code $n_{\mathrm{total}}$ and $n$ values to the ansatz indicated, and show the fit coefficients $A$, $B$, $C$ and $D$. In each case shown, $C < 2k$ and $A < k$, therefore we expect vine code efficacy over the surface code to only improve for larger distances. 
    \textbf{(a)} Rotated geometry. \textbf{(b)} Square geometry. \textbf{(c) } Vertical parallelogram.}
    \label{fig:nSeWwSs_results_fig2}
\end{figure*}

\begin{figure*}[ht]
    \centering
    \includegraphics{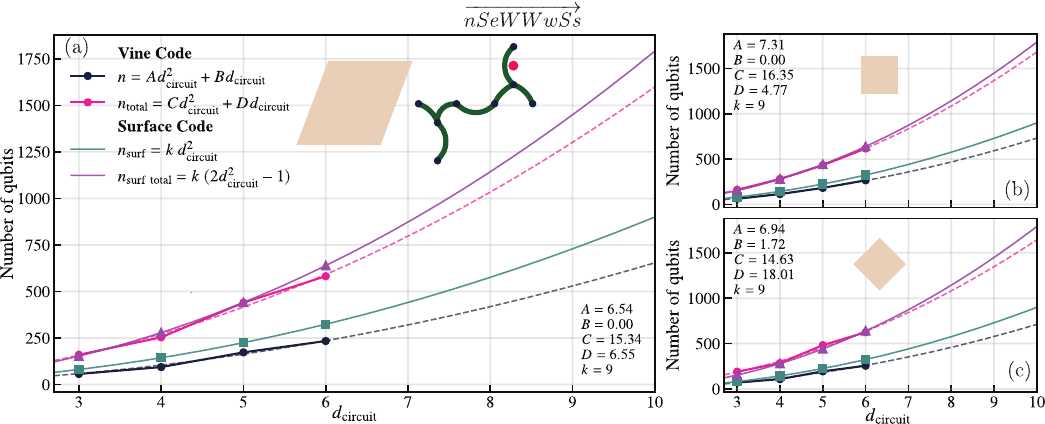}
    \caption{The qubit overhead scaling for three different open boundary patches of the $\overrightarrow{nSeWWwSs}$ vine code, each encoding $k = 9$ logical qubits. 
     \textbf{(a)} Horizontal parallelogram. \textbf{(b)} Square. \textbf{(c)} Rotated.}
    \label{fig:nSeWWwSs_results_fig3}
\end{figure*}

\section{Vine Codes}
\label{sec:vinecodes}

We first introduce the vine codes on an infinite plane and introduce boundaries in \cref{sec:boundaryconstruction}. Vine codes are defined with qubits on a planar square grid and with nearest neighbour gate connectivity, see \cref{fig:vine_codes}. Qubits are arranged so that all measure qubits have only data qubits as nearest neighbours. Inspired by Ref.~\onlinecite{Geher2025Directional}, vine codes are defined by a ``step sequence" which encode a series of gates, and a ``layout" which encodes relative initial positions of $X$- and $Z$- type measure qubits.

The step sequence specifies the series of gates that are performed during the syndrome extraction circuit. For simplicity, we construct steps as CXSWAP or CX gates (which are equivalent to iSWAP and CZ respectively, up to single-qubit Clifford gates). We define eight types of steps that can appear in a sequence. The steps $\vec{N}$, $\vec{E}$, $\vec{S}$, $\vec{W}$ each correspond to a CXSWAP gate performed between each measure qubit and the neighbouring data qubit in the North, East, South or West direction, see \cref{fig:vine_codes}(a)-(b).  Extending from Ref.~\onlinecite{Geher2025Directional}, we also introduce steps $\vec{n}$, $\vec{e}$, $\vec{s}$, $\vec{w}$, which correspond to a CX gate (no swap) performed between measure and data qubits in the indicated cardinal direction, see \cref{fig:vine_codes}(c)-(d). The control and target qubits of each CX/CXSWAP gate depend on whether the measure qubit involved in the gate is a $Z$ or $X$ measure qubit -- $Z$ ($X$) measure qubits are always the target (control) qubits of any entangling gates. 
For example, for an $\vec{N}$ step in the sequence, a CXSWAP controlled on an $X$ measure qubit will target the neighbouring data qubit directly above it in the square lattice grid.
Similarly, during the $\vec{N}$ step, CXSWAP gates that are targeted on $Z$ measure qubits will be controlled by data qubits directly above.

The supports of stabilisers follow a vine shape, with a main root generated by CXSWAP directions, and small offshoots generated by CX gates performed along the way. \cref{fig:vine_codes}(a)-(d) shows how black data qubits and pink or blue measure qubits are paired up and entangled for each of the operations defined. 

The layout encodes the initial positions of $X$ and $Z$ measure qubits. In \cref{fig:vine_codes}, we depict one typical layout which has alternating rows of $X$- and $Z$-stabilisers. This corresponds to ``layout 1" in Ref.~\cite{Geher2025Directional}. Many other valid layouts exist; for example we consider an alternative layout for a flip-vine code in \cref{sec:gen_bdrys}. However unless stated otherwise, we assume the layout in \cref{fig:vine_codes} for the remainder of this manuscript.

Figures~\ref{fig:vine_codes}(e)-(f) show an example vine code $\overrightarrow{nSewSs}$ on layout 1 and the corresponding vine code stabiliser support generated. The red measure qubit begins by entangling with the data qubit $d_1$ via a CX gate in the $\vec{n}$ direction. It does not move during this step. Next, we perform an $\vec{S}$ step, where the measure qubit is entangled with $d_2$ via a CXSWAP and therefore also swaps positions with $d_2$. This now allows the measure qubit to have nearest neighbour support with a different set of qubits compared to its initial position. The remaining moves, where the measure qubit does a combination of CX and CXSWAP gates, allow it to entangle with qubits $d_3$-$d_6$.
\cref{fig:vine_codes}(f) shows the support of a stabiliser and expanding detecting region~\cite{McEwen_2023_relaxing} generated by the $\overrightarrow{nSewSs}$ step sequence, relative to the dynamic position of the red measure qubit.

Importantly, in even syndrome extraction rounds, the measure qubit takes the opposite route, returning it to its initial position while entangling with data qubits in reverse order. The period of the total measure qubit walk returning to its original position 
is therefore two syndrome extraction rounds. While this return route convention is assumed for vine codes, when we introduce flip-vine codes in \cref{sec:flip_vine_codes} we introduce an alternative return instruction for even syndrome extraction rounds. 

Only certain step sequence and layout combinations define a valid stabiliser code replete with a syndrome extraction circuit. The combination needs to ensure that 1) stabilisers of opposite type have even overlapping support, 2) measure qubits are not entangled or ``interleaved" at the measurement step, 3) the code has topological order and 4) the code generates entanglement across the whole patch. The second condition requires that, for $X$- and $Z$-stabilisers with common data support, the subset of data qubits entangled with the $X$-measure qubit before the $Z$-measure qubit must be of even size, and vice versa. We discuss the conditions 1-4 further in \appref{app:valid_sequences}. Many vine code sequences are also equivalent circuits modulo $D_4$ symmetry, or generate stabilisers with identical data qubit support via different step sequences. We discuss these symmetries in \appref{app:exhaustive_sequence_search}.


\begin{figure*}
\centering
\begin{tikzpicture}
  \node[anchor=south west, inner sep=0] (img) at (0,0)
    {\includegraphics[width=0.8\linewidth]{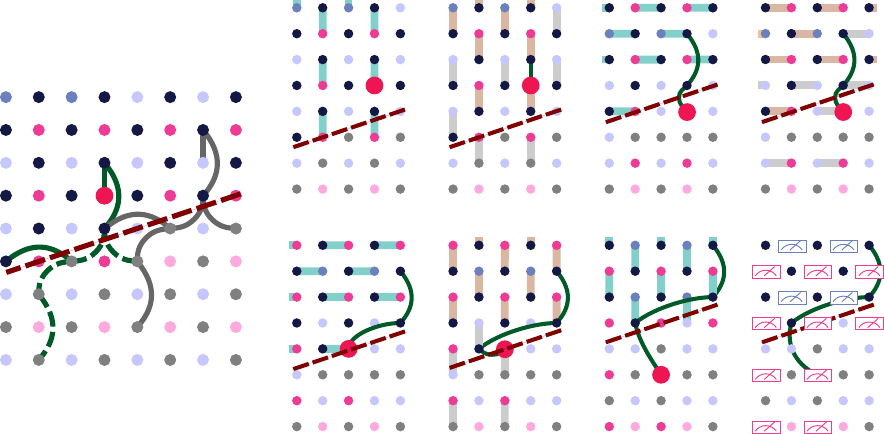}};
  \begin{scope}[x={(img.south east)}, y={(img.north west)}]
    \node[anchor=north west] at (-0.04, 0.850) {(a)};
    \node[anchor=north west] at (0.29, 1.00) {(b)};
    \node[anchor=north west] at (0.47, 1.00) {(c)};
    \node[anchor=north west] at (0.64, 1.00) {(d)};
    \node[anchor=north west] at (0.82, 1.00) {(e)};
    \node[anchor=north west] at (0.29, 0.45) {(f)};
    \node[anchor=north west] at (0.47, 0.45) {(g)};
    \node[anchor=north west] at (0.64, 0.45) {(h)};
    \node[anchor=north west] at (0.82, 0.45) {(i)};
  \end{scope}
\end{tikzpicture}
    \caption{An example of an $X$-type boundary for the $\overrightarrow{nSeWwSs}$ code (red dashed line). Black dots indicate data qubits, solid blue/pink dots indicate $Z$/$X$ measure qubits that are used for stabiliser measurements and remaining light grey/blue/pink dots indicate routing qubits that originated as data/$Z$/$X$ qubits before introducing the boundary. \textbf{(a)} An $X$-type stabiliser that intersects the boundary is truncated to its support on data qubits above the boundary (green vine, dashed below the boundary). A $Z$-type stabiliser that intersects the boundary is removed from the stabiliser group (grey vine). \textbf{(b)}-\textbf{(i)} The syndrome measurement circuit close to the boundary is shown. The green truncated $X$-type stabiliser is shown at each step. Light blue lines indicate CX gates, brown lines indicate CXSWAP gates, and grey lines indicate SWAP gates. The location of the boundary line follows the data qubits.\label{fig:boundary_stabilisers}}
\end{figure*}
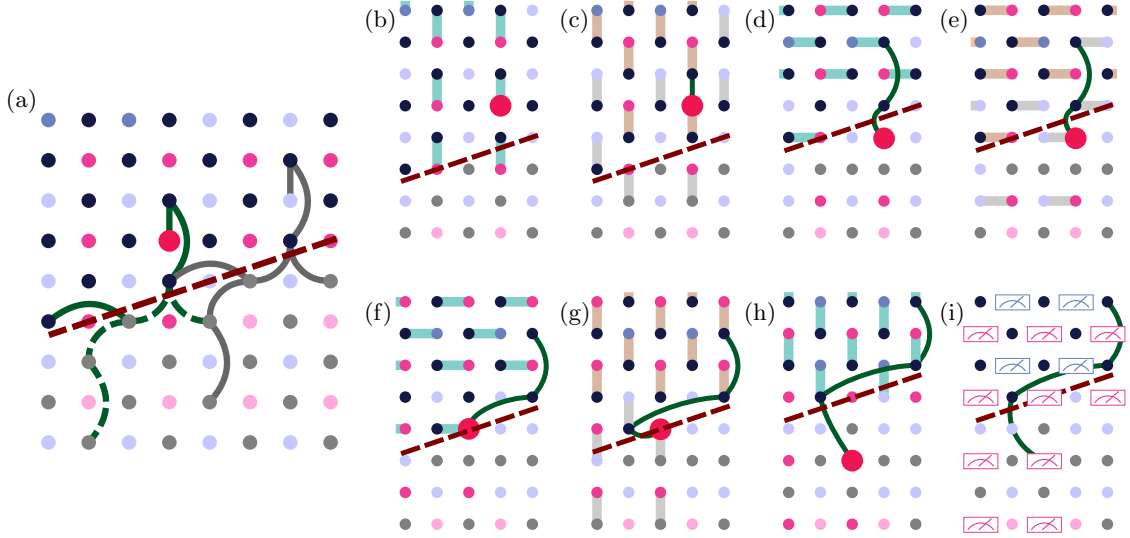

\subsection{Flip-vine codes for transversal Clifford gates}
\label{sec:flip_vine_codes}

Here, we propose a variant of vine codes that are weakly self-dual, meaning their $X$ and $Z$ parity check matrices are equivalent. A consequence of weak self-duality is that a Hadamard gate applied to all data qubits is a transversal logical operation. Additional constraints may similarly result in the presence of a transversal $S$ gate. An example of a family of codes with these properties is the 2D color codes~\cite{2DCC_bombin_delgado,Kubica_2015}.

The flip-vine codes are defined as follows. Firstly, select an even-length step sequence of length $w$ (and assume the stabiliser weights are also all $w$). We then select a subset of the measure qubits in the layout which will be ``active"; the remaining will be qubits used for routing (see below) and will not measure any stabilisers. All active measure qubits measure an associated stabiliser ($X$-type or $Z$-type) for odd QEC cycles. For even QEC cycles, when the step sequence is traced out in reverse, the measure qubits measure the opposite Pauli-type stabiliser with the same support ($Z$-type or $X$-type, respectively). In this way, the measure qubit's Pauli type ``flips" each QEC round, and for every stabiliser there is a corresponding stabiliser with opposite Pauli type with identical support. 

Given that measure qubits of flip-vine codes now support double the number of stabilisers relative to vine codes, half (or approximately half when open boundaries are constructed) the measure qubits of the flip-vine code will be active, while the remainder will be considered as routing qubits, performing swaps during their step sequence instead of entangling gates. This is to ensure that the total number of stabilisers is not increased and that all stabilisers commute.

Here, we prescribe how to determine which measure qubits remain active and which become routing in a flip-vine code. 
Let us compare the regular vine code with a given step sequence $\mathcal{V}$ to the flip-vine code with the same sequence $\mathcal{V}_F$. For $\mathcal{V}$, let us denote the set of $X$ ($Z$) measure qubits $A_X$ ($A_Z$) and the set of data qubits $\mathcal{D}$. 
In $\mathcal{V}_F$, active measure qubits $\mathcal{A}$ produce stabiliser generators of both Pauli types, and routing qubits $\mathcal{R}$ do not entangle at all with the code. 

To decide which measure qubits should be demoted to routing qubits, we start with an arbitrary measure qubit $a_0$ which we define to be active and consider all the stabilisers of the vine code, $\mathcal{V}$. Let measure qubit $a_0$ be associated with stabiliser $S_{a_0}$. Consider the $n$ stabilisers that have odd overlap with $S_{a_0}$; these are associated with measure qubits separated from $a_0$ by a vector from a set we denote $\Delta'_\text{odd}$, defined in Appendix~\ref{app:delta_sets}. Let us label the $n$ measure qubits associated to these stabilisers $\{a_1,a_2,\ldots, a_n\}$. In $\mathcal{V}_F$, if these measure qubits are active, they support both an $X$-type and a $Z$-type stabiliser, both of which will anti-commute with a stabiliser associated with $a_0$. Hence, because $a_0$ is active, $\{a_1,a_2,\ldots, a_n\}$ must be demoted to routing qubits in $\mathcal{V}_F$ to ensure all stabilisers commute.

Let us then construct the total lattice of routing qubits that must include these $\{a_1,a_2,\ldots a_n\}$ in a systematic way. To do so, we start with an arbitrary ``seed" routing qubit position $\text{seed}_X\in A_X$. We then consider all qubits at positions  $\text{seed}_X + \overrightarrow{u_1}+\overrightarrow{u_2}$, where $\overrightarrow{u_1},\overrightarrow{u_2}\in \Delta'_\text{odd}$, as routing qubits also. Subsequently, we can build the lattice of routing qubits centered at seed$_X$, adding in all qubits connected to seed$_X$ by any even sum of $u_i\in \Delta'_\text{odd}$ vectors. Following the same logic, we define a set of routing qubits that start from a $\text{seed}_Z \in A_Z$. This results in the creation of two disjoint lattices $\mathcal{R}_{X/Z} = \lbrace \text{seed}_{X/Z} + \sum_{i=1}^{2\ell}\overrightarrow{u_i} \, | \, \ell \in \mathbb{Z},\, \overrightarrow{u_i}\in \Delta'_\text{odd}\rbrace$.  From here, we define $\mathcal{R}$ to be the union of all qubits in $\mathcal{R}_X$ and $\mathcal{R}_Z$. Given vectors from $\Delta'_\text{odd}$ connect measure qubits of the same type in $\mathcal{V}$ by definition (see Appendix~\ref{app:delta_sets}), $\mathcal{R}_X\subset A_X$ and $\mathcal{R}_Z \subset A_Z$. Assuming that there exist no $\overrightarrow{u_i}\in \Delta'_\text{odd}$ such that $\sum_{i=1}^{2\ell}\overrightarrow{u_i}\notin \Delta'_\text{odd}$ for some integer $\ell$ (true for all cases we consider in this manuscript), then all remaining measure qubits can support $Z$ and $X$ stabilisers that mutually commute, and hence are left as active measure qubits.

The circuit for the flip-vine code is then defined as follows: measure qubits in $A_X\backslash \mathcal{R}_X$ ($A_Z\backslash \mathcal{R}_Z$) measure an $X$-type ($Z$-type) stabiliser in odd QEC rounds, and a $Z$-type ($X$-type) stabiliser with the same support in even QEC rounds. The measure qubits follow the step sequence defining the flip-vine code in odd rounds and the reverse sequence in even rounds. Routing qubits in $\mathcal{R}_X$ and $\mathcal{R}_Z$ traverse the same ``vine" defined by the step sequence as the active measure qubits, but with SWAP gates rather than CXSWAP gates and therefore do not entangle with data qubits. We can use superconducting-native two-qubit gates and resets to perform SWAPs in a single circuit layer, see \cref{sec:pauli_bdrys}. Data qubits traverse the reverse step sequence to the measure/routing qubits. For some of the circuit layers in a single syndrome extraction round (forwards or reverse), a given data qubit will be involved in entangling gates, while for the rest of the circuit layers, it will simply be swapping with routing qubits. 

With a suitable choice of boundary conditions -- either periodic boundary conditions, or certain generalised open boundary conditions (see Section~\ref{sec:gen_bdrys}) -- the resulting code is weakly self-dual. All flip-vine codes (with appropriate boundary conditions) therefore support a transversal Hadamard gate. The logical action of this gate will be $\overline{H}_j$ for all $j=1,\ldots, k$. The codes will also possess a transversal $S$ gate if we can find a bipartition of the data qubits into subsets $A$ and $B$ such that $|s\cap A| \equiv |s\cap B|$ mod $4$, for all stabiliser data qubit support sets $s$. We may then apply $S$ ($S^\dagger$) gates to all data qubits in $A$ ($B$) to obtain a transversal gate~\cite{Kubica_2015}. This is because an $X$ stabiliser conjugated by such a gate maps to a product of an $X$ and $Z$ stabiliser, up to a phase. This phase is equal to $i^{|s\cap A|-|s\cap B|}$. All codes we find admit such a bipartition. Note that all weight-8 stabiliser codes admit a trivial bipartition of $ A = \mathcal{D}$ and $B = \emptyset$. Therefore, stabiliser weight-$8$ codes admit a logical $S$ gate by simply applying $S$ transversally to all data qubits. The logical action is $\overline{S}_j$ for $j=1\ldots k$.

\section{Boundary Construction}
\label{sec:boundaryconstruction}

Here, we demonstrate how to produce explicit finite-size instances of vine codes with open boundaries. We begin by describing the construction of $X$/$Z$ boundaries which we refer to as Pauli boundaries, before discussing a more general approach.

\begin{figure*}
    \centering
    \includegraphics{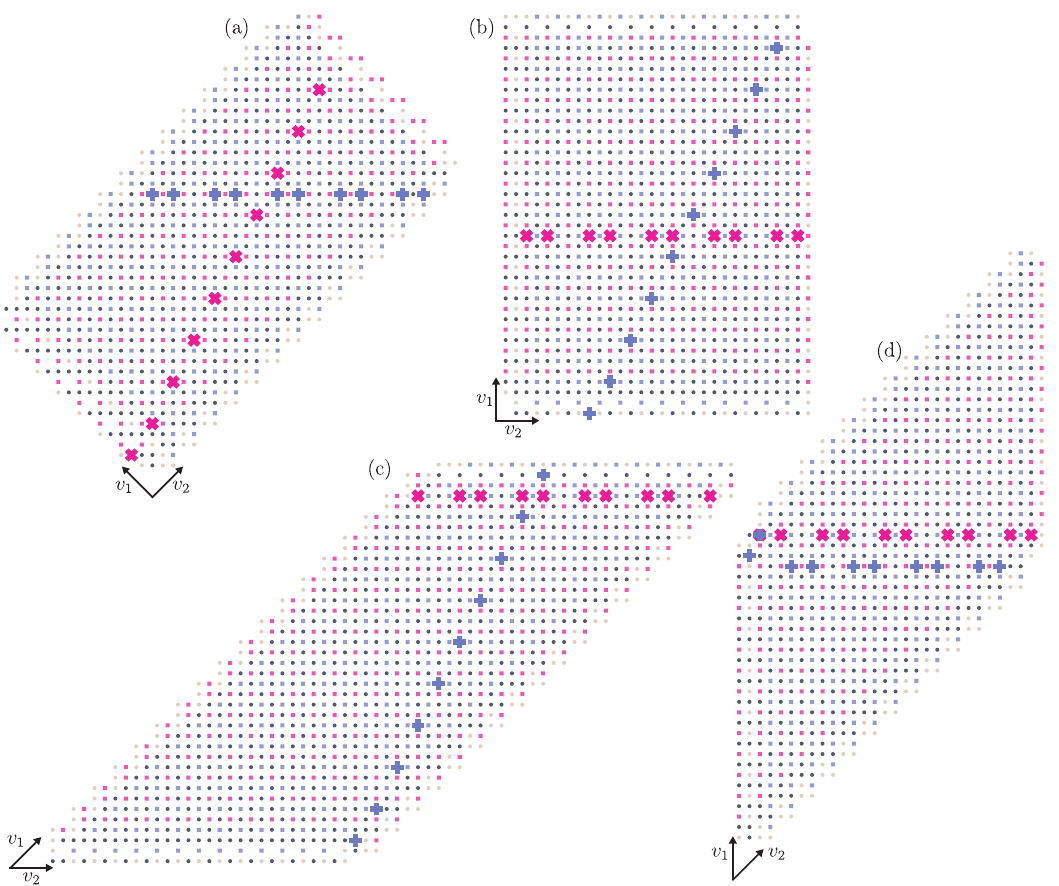}
    \caption{The $\overrightarrow{nSeWwSs}$ code with $d = 10$ for four boundary geometries: rotated, square, horizontal parallelogram and vertical parallelogram. Each example has the same boundary types, i.e. an $X$-type ($Z$-type) along the boundaries parallel to $v_1$ ($v_2$). The boundary geometry given by the unit vectors of $v_1$ and $v_2$ in each case has a significant effect on the number of qubits required to achieve $d=10$ (see \cref{tab:distance_10_codes}). 
    We show an $X$- ($Z$)-logical representative for each patch with pink crosses (blue crosses). 
    Lower-distance patches showed no hook errors for the same four geometries ($d=d_{\mathrm{circuit}}$; see \cref{tab:performing_boundary_cuts}).
    Therefore, we would expect the logical supports presented here to also represent the lowest-weight logicals at the circuit level. \textbf{(a)} Rotated geometry.
    \textbf{(b)}  Square geometry. 
    \textbf{(c)} Horizontal parallelogram.
    \textbf{(d)} Vertical parallelogram. 
    \label{fig:code_patches}}
\end{figure*}

\subsection{Pauli boundaries}\label{sec:pauli_bdrys}

We construct $X$ and $Z$ Pauli boundaries where $X$ and $Z$ logicals terminate respectively. These boundary types are also constructed in tile codes~\cite{steffan2025tilecodes,liang2025planarquantumlowdensityparitycheck}, and rotated and unrotated surface codes. Similarly to those cases, we expect vine codes with Pauli boundaries to encode half the number of logical qubits as the maximum possible with periodic boundary conditions~\cite{Liang_2025_gen_toric_codes}. 

To construct the boundary geometry, we first define vectors $v_1$ and $v_2$ that define a parallelogram with corners $\lbrace 0, v_1,v_2,v_1+v_2\rbrace$. We envisage cutting out data qubits supported on this parallelogram patch from the infinite plane.
We will construct boundaries running from $0$ to $v_1$ and from $v_2$ to $v_1+v_2$ to be Pauli $X$ boundaries, and those running from $0$ to $v_2$ and $v_1$ to $v_1+v_2$ to be Pauli-$Z$ boundaries. 

To do this, we first consider the data qubit support of all stabiliser operators. We firstly discard all $X$- and $Z$-stabilisers that have support fully outside the parallelogram. For stabilisers with partial support inside the parallelogram, we inspect whether the stabiliser intersects a Pauli $X$ or $Z$ boundary. We discard all $X$($Z$) stabilisers that intersect a Pauli-$Z$($X$) boundary. Meanwhile, we truncate every $X$($Z$) stabiliser that intersects a Pauli-$X$($Z$) boundary to the partial support within the parallelogram. Figure~\ref{fig:boundary_stabilisers}(a) shows an example of a discarded $Z$ stabiliser (grey) intersecting an $X$ boundary (red dashed line), along with an $X$ stabiliser truncated to a weight-3 operator, as its support below the boundary is removed (green solid/dashed line). We discuss the treatment of corners in Appendix~\ref{app:corners}. We then delete the following qubits:
\begin{itemize}
    \item All data qubits that lie outside the patch parallelogram
    \item All measure qubits that measure a removed stabiliser 
    \item All data qubits (within the parallelogram patch) that are not in the support of a stabiliser of both Pauli-types (such cases would introduce low-weight logicals).
\end{itemize}
Finally, this procedure may result in a code whose Tanner graph contains multiple disjoint connected components. For example, there may be isolated Bell pairs at corners that are not entangled with the rest of the patch due to removed stabilisers. We  keep all data/measure qubits that form the largest Tanner graph connected component. In some cases, the vine code sequence does not generate entanglement across the whole system even in the infinite plane and the boundary procedure described in this section results in approximately equal sized disconnected regions. See \appref{app:valid_sequences} for a discussion of this case; we ignore these instances in our exhaustive vine code search in \cref{sec:vine_code_search}.

The above procedure gives us the measure and data qubits needed for a patch with Pauli boundaries for a given parallelogram. However, given some stabilisers near the boundary have truncated support, their associated measure qubits do not need to entangle with all the qubits in their step sequence as some data qubits that originally featured in their support have been discarded. Similarly, some data qubits near the boundary will not need to be entangled with the measure qubits along their step sequence that have been discarded.
However, these discarded measure and data qubits may still be required in the planar grid to achieve the required connectivity. We label these ``routing qubits". The simplest solution that achieves the required connectivity maintains that all data and measure qubits walk along their usual trajectories, including near the boundary. If a data/measure qubit $q$ needs to step in direction $\vec{D}$ but there is no measure/data qubit located in that direction because it has been discarded by the process above, we define a routing qubit in that location. We then exchange the entangling gate for a SWAP gate in the corresponding circuit layer between $q$ and the routing qubit. 

The circuit that measures the truncated stabilisers close to the boundary using this method is straightforward. One applies the same step sequence to all measure qubits (and the reverse sequence to all the data qubits). For each measure qubit $a$ and each step in the sequence, $\vec{D}$, if there is a data qubit in the $\vec{D}$ direction from $a$, the usual gate at that point in the sequence is applied (CX or CXSWAP). Otherwise, if there is a routing qubit in that direction, a SWAP gate is applied if the step in the sequence is in $\lbrace \vec{N}, \vec{E}, \vec{S}, \vec{W}\rbrace$, or otherwise no gate is applied. Figure~\ref{fig:boundary_stabilisers}(b)-(i) shows an example of a truncated $X$ stabiliser being measured by the syndrome extraction circuit.

We can perform these SWAP gates in a single circuit layer by resetting the routing qubits in state $\ket{0}$ at the beginning of each QEC round and applying CZSWAP gates instead of SWAP gates between routing and other qubits. In the absence of errors, the effect of this gate is the same as a SWAP gate. We can also measure the routing qubits at the end of each QEC round and form a ``flag" detector from this measurement -- a flipped measurement outcome of $-1$ indicates an error has occurred~\cite{Chamberland_2020}.

In many cases, we may reduce the routing qubit number by removing unnecessary SWAP gates from the circuit near the boundaries. We use a series of techniques (see Appendix~\ref{app:swap_heuristics}) to find these removable SWAP instances and observe large reductions in required routing qubit numbers by removing them. For example, the distance-10 $\overrightarrow{nSeWwSs}$ code patch in Fig.~\ref{fig:code_patches}(a) has its routing qubit number reduced by $\sim 66.2\%$ (from $198$ to $67$; see \cref{tab:distance_10_codes}).

Figure~\ref{fig:code_patches} shows examples of parallelogram code patches obtained for the $\overrightarrow{nSeWwSs}$ step sequence using the procedure above. We consider four parallelogram geometries which we refer to as rotated, square, horizontal parallelogram, and vertical parallelogram. All of these examples support distance-10 codes (for both $X$ and $Z$ logical operators) and we provide example minimum-weight logical operators for each patch. While using the same vine code sequence to generate a code encoding the same number of logical qubits $k=6$ at distance $10$, the qubit overheads vary dramatically (see \cref{tab:distance_10_codes}). In this example, the rotated patch has the smallest qubit overhead, requiring $82\%$ of the qubits required for equivalent surface code patches. The horizontal parallelogram has the largest overhead, requiring $104\%$ of the total qubits required for the surface code.

\begin{figure*}
    \centering
    \includegraphics[width=0.99\linewidth]{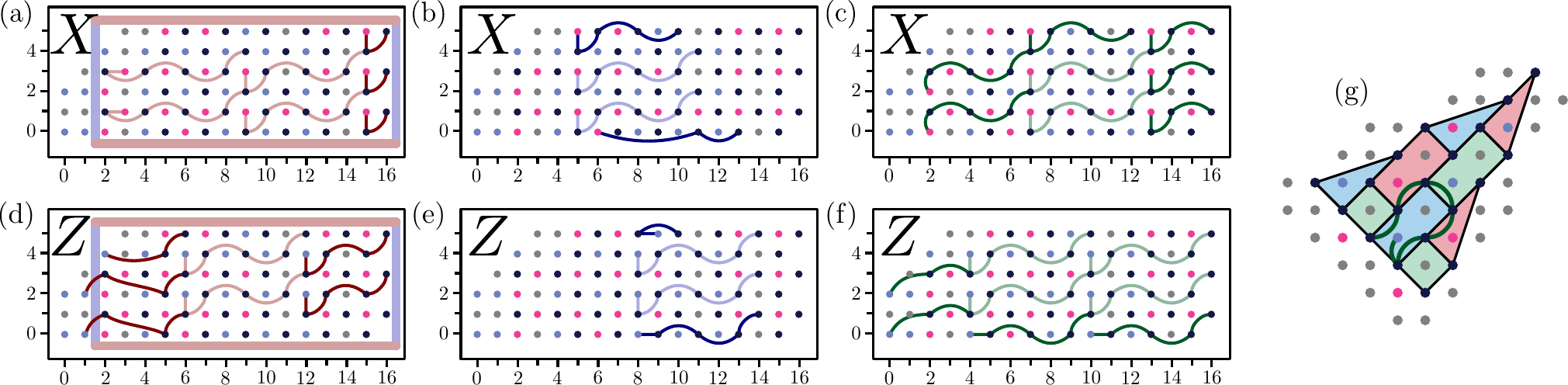}
    \caption{Examples of coloured boundaries applied to (flip-)vine codes that are equivalent (up to a finite-depth circuit) to the 2D color code. Data/$Z$-measure/$X$-measure/routing qubits are coloured black/blue/pink/grey, respectively. In \textbf{(a)}-\textbf{(f)} we present the stabilisers of a patch of $\overrightarrow{sEEEn}$ code (equivalent in support to $\overrightarrow{NEEEN}$ from Ref.~\cite{Geher2025Directional}) with coloured boundaries. In (a)-(c) [(d)-(f)] we present the $X$ ($Z$) stabilisers, all drawn between measure and data qubits in their starting positions. We colour the stabilisers of both types either red [(a) and (d)], blue [(b) and (e)] or green [(c) and (f)], and we colour the vertical (horizontal) boundaries blue (red), as shown in (a) and (d). Those stabilisers that intersect a boundary of the same (a different) colour are discarded (truncated by tunnelling operators). Any stabiliser that has been truncated is highlighted, while the bulk stabilisers are lighter in colour. In \textbf{(g)}, we present a triangular patch of $\overrightarrow{sENWSw}$ (equivalent sequence to $\overrightarrow{NENWSW}$) flip-vine code. The stabiliser support and colour is indicated by the coloured plaquettes, each hosting one $Z$ and one $X$ stabiliser. The colour of the measure qubit associated to a plaquette is given by the Pauli type of the stabiliser measured in odd rounds (in even rounds, the opposite Pauli type stabiliser is measured). Three of the plaquettes (the hexagons) are in the bulk, and the rest host boundary stabilisers. All stabilisers are drawn between the measure and data qubits as they sit at the start of the circuit. See \appref{app:coloured_bdrys} for the full circuits.\label{fig:coloured_bdry_stabs}}
\end{figure*}

\subsection{Generalised boundaries}\label{sec:gen_bdrys}

We now describe boundary constructions beyond Pauli boundaries. While we will not cover boundaries in full generality, we will introduce relevant theory for generalised boundaries along with several practical examples.

In this section, we consider codes that are equivalent (up to finite-depth unitaries) to the 2D color code
to construct ``coloured boundaries" that are analogous to the red/green/blue boundaries of the color code.~\cite{Kesselring_2018_bdrys_twists,Kesselring_2024_anyon_condensation}.
Since vine codes are 2D-local, translationally invariant, commuting Pauli stabiliser codes defined on qubits, 
they are equivalent (up to finite-depth unitary transformations) to $N \in \mathbb{Z}$ copies of the toric code~\cite{Bombin_2012_universal}. The integer $N$ can be obtained from the number of logical qubits the vine code supports when defined with Pauli boundaries (on a large enough patch). Therefore, our starting point is vine codes that encode $k = 2$ logical qubits on a patch with Pauli boundaries (which includes codes with $k = 4$ on a torus in Ref.~\cite{Geher2025Directional}). These vine codes will instantiate the same topological phase as the 2D color code (see \appref{app:coloured_bdrys} for more information on the 2D color code boundaries and anyons).

The first step is to identify the ``anyon model" that describes the syndrome defect equivalence classes of the code; this requires associating an ``anyon type"  to each stabiliser generator. We associate a combination of toric code charges to each anyon type.
Specifically, each anyon type $a$ associated with an $X$ stabiliser ($b$ associated with a $Z$ stabiliser) is a product of toric code layer charges $a=e_1^{c_1}e_2^{c_2}\ldots e_N^{c_N}$ ($b=m_1^{c_1}m_2^{c_2}\ldots m_N^{c_N}$), for $c_i\in \lbrace 0, 1\rbrace$ where $i=1,\ldots, N$ labels the toric code layer.
Let us define the set of ``tunnelling operators" $\mathcal{X}_j^{a}$ ($\mathcal{Z}_j^{b}$), with each element indexed by an integer $j$, to be the set of Pauli $X$- ($Z$-)operators that flip exactly \textit{two} stabilisers associated with a particular anyon type $a$ ($b$). This operator is therefore associated with transporting a syndrome charge without increasing the weight of the total syndrome. Similarly, logical operators can be constructed by combining tunnelling operators of the same anyon type.

To construct a boundary, we consider moving from an infinite-plane code to a half-infinite plane code supported only in the $x\geq 0$ region. In the $x\leq 0$ region, we construct a trivial vacuum phase. To construct the domain wall to the vacuum phase, we first identify a maximal set of commuting tunnelling operators and applying projectors for those tunnelling operators fully suppported in the $x\leq 0$ region. For example, if $\mathcal{X}^{a}_j$ and $\mathcal{Z}^b_j$ associated with anyons $a$ and $b$ form a maximal set of commuting tunnelling operators we apply projectors $\frac{1}{2}(1+\mathcal{X}^{a}_j)$ and $\frac{1}{2}(1+\mathcal{Z}^b_j)$.  In this way,  $a$ and $b$ now belong to a ``Lagrangian subgroup" of the anyon model (for more information, see Refs.~\onlinecite{Kesselring_2018_bdrys_twists,Kesselring_2024_anyon_condensation,Kapustin_2011,Levin_2013,Barkeshli13b,Yoshida_2017}). Equivalently, the anyons  $a$ and $b$ are ``condensed" at this boundary and their associated logical operators can terminate here.

Applying projectors in this way corresponds to ``measuring out" the associated commuting tunnelling operators that are confined to the $x\leq 0$ region. We then discard any vine code stabilisers that anti-commute with the measured tunnelling operators, and use these new tunnelling operator stabilisers to truncate all other vine code stabilisers to within the boundaries of the code patch. Note that this framework captures the Pauli boundaries introduced above, where the tunnelling operators for a given boundary are all of either $X$ or $Z$ type. 

In the case of vine codes supporting an equivalent anyon model to the color code, we can colour the stabiliser generators and associated anyons red, green, or blue. In this way, the tunnelling operators for stabilisers of the same colour mutually commute (see Appendix~\ref{app:coloured_bdrys}). Coloured boundaries condense both the $X$ and $Z$ anyons associated with a single colour label. 

We consider two explicit examples for coloured boundary construction: the $\overrightarrow{sEEEn}$ vine code, and the $\overrightarrow{sENWSw}$ flip-vine code. Both of these encode $k = 2$ qubits on a Pauli boundary patch.  In Fig.~\ref{fig:coloured_bdry_stabs}, we show the result of creating small code patches with coloured boundaries for these two codes.  In Fig.~\ref{fig:coloured_bdry_stabs}(a)-(f), we show the $X$ and $Z$ stabiliser supports for the $\overrightarrow{sEEEn}$ code with vertical ``blue" boundaries and horizontal ``red" boundaries.  To construct boundaries of a given colour, say blue, we then discard any blue stabilisers (of either $X$ or $Z$ type) that intersect the boundary, and truncate red and green stabilisers that intersect the boundary using blue tunnelling operators. This can produce some complexities, which we discuss in Appendix~\ref{app:coloured_bdrys}. We finally construct a valid syndrome measurement circuit around the coloured boundaries, which is a non-trivial process. This involves gates that break the step sequence pattern near the boundaries. We present the full circuits in Appendix~\ref{app:coloured_bdrys}. We then obtain a patch with stabilisers shown in Fig.~\ref{fig:coloured_bdry_stabs}(a)-(f), with all data/measure qubits in their starting locations for the syndrome measurement circuit. This patch encodes $k=2$ logical qubits.

For the flip-vine code shown in Fig.~\ref{fig:coloured_bdry_stabs}(g), we can obtain an exact copy of a 2D color code patch. In the figure, we show an example with a triangular shape that encodes a single logical qubit. Note that the step sequence does not produce a valid code on the layout presented in \cref{fig:vine_codes}, and hence we use an alternative initial arrangement of $X$ and $Z$ measure qubits as shown in \cref{fig:coloured_bdry_stabs}(g) (referred to as layout 3 in Ref.~\onlinecite{Geher2025Directional}). Stabilisers of $X$ and $Z$ type with the same support retain equal supports after truncation (indicated by coloured plaquettes in the figures). The circuit that measures these stabilisers simply involves utilising routing qubits for these truncated stabilisers, in the same way as occurs when Pauli boundaries are imposed (see Section~\ref{sec:pauli_bdrys}). Owing to the coloured boundaries being Pauli-symmetric, the resulting code is still weakly self-dual and supports transversal single-qubit Clifford gates, as is well known for the color code~\cite{2DCC_bombin_delgado}.

\section{Results}\label{sec:results}

Vine codes show promising performance over the surface code. Here, we do an exhaustive search to find all valid vine codes up to stabiliser weight $9$, see \cref{sec:vine_code_search}. We benchmark the performance of selected vine code in multiple ways. First, we calculate their qubit overhead scaling, finding efficiencies relative to the surface code, see \cref{fig:nSeWwSs_results_fig2}, \cref{fig:nSeWWwSs_results_fig3}, \cref{fig:code_patches}, \cref{tab:performing_boundary_cuts} and \cref{sec:qubit_overhead_scaling}. We corroborate their performance by simulating their syndrome extraction circuits under circuit-level noise, see \cref{fig:sim_results} and \cref{sec:circuit-level_noise_simulations}. Finally, we present examples of flip-vine codes that permit transversal single-qubit Clifford gates, showing examples in \cref{tab:flip_vine_table}, see \cref{sec:flip-vine_code_results}.

\subsection{Vine code search}
\label{sec:vine_code_search}

We perform exhaustive searches over all possible vine code sequences up to stabiliser weight $9$. We filter for sequences that have: commutativity, independently measured stabilisers, topological order, and full patch connectivity, see \appref{app:valid_sequences} for a full description of conditions. We then reduce the number of sequences by by removing any that have step sequences or supports related by a rotation/reflection symmetry from $D_4$,
see \appref{app:exhaustive_sequence_search}. We show a final list of vine code sequences in \cref{tab:topological_vine_codes_wt1-9}.

We find $147$ step sequences up to and including stabiliser weight $9$. Of those, $96$ are stabiliser weight $9$, and $51$ are stabiliser weight $8$ or less. We focus on weight $\leq 8$ sequences in \cref{sec:qubit_overhead_scaling}.

Note that before we filter vine codes with equivalent data qubit support (see step \textbf{7} in \appref{app:exhaustive_sequence_search}), we have $906$ vine codes up to stabiliser weight $9$. In \cref{tab:topological_vine_codes_wt1-9}, we choose a single representative. However, vine codes with different entangling sequences that produce the same data qubit support may differ in circuit distance, due to possessing different hook errors. We do not optimise for hook errors when choosing a canonical representative.

\begin{table*}
\centering
\renewcommand{\arraystretch}{1.2}
\begin{tabular}{ccccccccc}
\toprule
\textbf{Vine Sequence} & \textbf{Cut} & $v_1$ & $v_2$ & \textbf{[[$n$, $k$, ($d_x$, $d_z$)]]} & $d_{\mathrm{circuit}}$ & $n_{\mathrm{total}}$ & $\frac{n_{\mathrm{total}}}{n_{\mathrm{surf\ total}}}$ & $\frac{n_{\mathrm{data+meas}}}{n_{\mathrm{surf\ total}}}$ \\
\midrule
\multirow{4}{*}{$\overrightarrow{nSeWwSs}$} & square & {[}0, 22{]} & {[}16, 0{]} & $[[180, 6, (6, 6)]]$ & 6 & 393 & 0.9225 & 0.8310 \\
 & rotated & {[}-10, 10{]} & {[}20, 20{]} & $[[221, 6, (7, 7)]]$ & 7 & 511 & 0.8780 & 0.7491 \\
 & horizontal parallelogram & {[}22, 22{]} & {[}16, 0{]} & $[[191, 6, (6, 6)]]$ & 6 & 456 & 1.0704 & 0.8826 \\
 & vertical parallelogram & {[}0, 20{]} & {[}20, 20{]} & $[[231, 6, (7, 7)]]$ & 7 & 535 & 0.9192 & 0.7835 \\
\midrule
\multirow{4}{*}{$\overrightarrow{nSeWWwSs}$} & square & {[}0, 22{]} & {[}24, 0{]} & $[[266, 9, (6, 6)]]$ & 6 & 617 & 0.9656 & 0.8185 \\
 & rotated & {[}-11, 11{]} & {[}23, 23{]} & $[[256, 9, (6, 6)]]$ & 6 & 625 & 0.9781 & 0.7872 \\
 & horizontal parallelogram & {[}22, 22{]} & {[}22, 0{]} & $[[234, 9, (6, 6)]]$ & 6 & 582 & 0.9108 & 0.7183 \\
 & vertical parallelogram & {[}0, 22{]} & {[}24, 24{]} & $[[278, 9, (6, 6)]]$ & 6 & 683 & 1.0689 & 0.8560 \\
\midrule
\multirow{2}{*}{$\overrightarrow{nSENESs}$} & square & {[}0, 22{]} & {[}15, 0{]} & $[[173, 4, (7, 12)]]$ & 7 & 384 & 0.9897 & 0.8814 \\
 & rotated & {[}-14, 14{]} & {[}13, 13{]} & $[[182, 4, (7, 8)]]$ & 7 & 434 & 1.1186 & 0.9278 \\
\midrule
\multirow{1}{*}{$\overrightarrow{nEeSSwEs}$} & rotated & {[}-17, 17{]} & {[}12, 12{]} & $[[207, 7, (6, 6)]]$ & 6 & 500 & 1.0060 & 0.8189 \\
\midrule
\multirow{1}{*}{$\overrightarrow{nENEENEs}$} & horizontal parallelogram & {[}13, 13{]} & {[}22, 0{]} & $[[144, 5, (6, 6)]]$ & 6 & 364 & 1.0254 & 0.7972 \\
\midrule
\multirow{2}{*}{$\overrightarrow{nESWWSEs}$} & square & {[}0, 22{]} & {[}13, 0{]} & $[[151, 5, (6, 6)]]$ & 6 & 372 & 1.0479 & 0.8394 \\
 & vertical parallelogram & {[}0, 22{]} & {[}13, 13{]} & $[[143, 5, (6, 6)]]$ & 6 & 373 & 1.0507 & 0.7915 \\
\midrule
\multirow{2}{*}{$\overrightarrow{nENENEs}$} & rotated & {[}-7, 7{]} & {[}17, 17{]} & $[[130, 4, (6, 8)]]$ & 6 & 320 & 1.1268 & 0.9014 \\
 & horizontal parallelogram & {[}17, 17{]} & {[}14, 0{]} & $[[121, 4, (8, 6)]]$ & 6 & 298 & 1.0493 & 0.8380 \\
\midrule
\multirow{1}{*}{$\overrightarrow{nSeWweSs}$} & rotated & {[}-11, 11{]} & {[}17, 17{]} & $[[206, 6, (6, 6)]]$ & 6 & 479 & 1.1244 & 0.9531 \\
\midrule
\multirow{1}{*}{$\overrightarrow{nENNEs}$} & vertical parallelogram & {[}0, 14{]} & {[}13, 13{]} & $[[105, 3, (6, 9)]]$ & 6 & 248 & 1.1643 & 0.9718 \\
\bottomrule
\end{tabular}
\caption{Boundary cut and qubit overhead information for selected vine sequences where the number of data and measure qubits ($n_{\mathrm{data+meas}}$) needed to encode $k$ logical qubits is less than what is required for $k$ surface codes with the same circuit level distance, i.e. $n_{\mathrm{data+meas}} / n_{\mathrm{surf\ total}} < 1$. Here, $v_1$ and $v_2$ give the vectors defining the boundary of the vine code cut into the shape as indicated (see also \cref{fig:code_patches}), $n$ gives the number of data qubits, $(d_x, d_z)$ are code distances in each basis, $d_{\mathrm{circuit}}$ is the circuit distance of the patch in both bases,  $n_{\mathrm{total}}$ is the total number of qubits used for the vine code, including data, measure and routing qubits, and $n_{\mathrm{surf\ total}}$ is the total number of data and measure qubits needed to encode $k$ logical qubits in the surface code (i.e. $n_{\mathrm{surf\ total}} = k (2 d_{\mathrm{circuit}}^2 - 1)$. 
While we show here only patches with circuit level distances verified by an exact distance finder~\cite{pryadko2024distm4ri}, we expect the efficiency of the vine code relative to the surface code to improve as distances increase (see \cref{fig:nSeWwSs_results_fig2}, and \cref{fig:nSeWWwSs_results_fig3}).  }
\label{tab:performing_boundary_cuts}
\end{table*}

\begin{table*}
\centering
\begin{tabular}{cccccccc}
    \textbf{Cut} & $v_1$ & $v_2$ & $[[n, k, d_\text{circuit}]]$ & $n_\text{routing}\rightarrow n_\text{reduced}$ & $n_\text{total}$ & $\frac{n_\text{total}}{n_{\mathrm{surf\ total}}}$ & $\frac{n_\mathrm{data+measure}}{n_\mathrm{surf\ total}}$ \\
    \midrule
    Rotated & $[-14,14]$ & $[29,29]$ & $[[437,6,10]]$ & $194\rightarrow 112$ & $980$ & $0.8208$ & $0.7270$\\
    \midrule
    Square & $[0,38]$ & $[28,0]$ & $[[538,6,10]]$ & $198 \rightarrow 67$ & $1137$ & $0.9523$ & $0.8961$\\
    \midrule
    Horizontal parallelogram & $[38,38]$ & $[28,0]$ & $[[557,6,10]]$ & $236\rightarrow 140$ & $1248$ & $1.0452$ & $0.9280$\\
    \midrule
    Vertical parallelogram & $[0,28]$ & $[28,28]$ & $[[435,6,10]]$ & $192\rightarrow 111$ & $975$ & $0.8166$ & $0.7236$
\end{tabular}
\caption{Qubit overheads relative to the surface code of the four patch geometries of the $\overrightarrow{nSeWwSs}$ code shown in \cref{fig:code_patches}, defined with vectors $v_1$ and $v_2$. In these patches, we conjecture $d_x=d_z = d_\text{circuit}$. We do not verify the circuit distances for these code patches, but we extrapolate from smaller patches that have no distance-reducing hook errors (see \cref{tab:performing_boundary_cuts}). Here, $n_\text{routing}$ is the number of routing qubits naively required for the patches, and $n_\text{reduced}$ is the reduced requirement after using techniques from \appref{app:swap_heuristics}. The final three columns use quantities defined in the caption of \cref{tab:performing_boundary_cuts}.}
\label{tab:distance_10_codes}
\end{table*}

\subsection{Qubit overhead scaling}
\label{sec:qubit_overhead_scaling}

We present selected vine codes and their qubit overhead scaling. We find small instances of vine codes with circuit distance $6$ or $7$ that already require fewer qubits than equivalent surface code patches. These efficiencies are also expected to increase as the circuit distance increases. 

Figure~\ref{fig:nSeWwSs_results_fig2} and \cref{fig:nSeWWwSs_results_fig3} show qubit overhead scaling for the $\overrightarrow{nSeWwSs}$ and $\overrightarrow{nSeWWwSs}$ codes. In each case, we construct patches with $X$/$Z$ open boundaries that obtain equal circuit distances in both bases, up to circuit distance $6$ or $7$. We verify circuit distances from an exact distance finder \cite{pryadko2024distm4ri}. We find the total number of qubits $n_{\mathrm{total}}$, including $n$ data qubits as well as measure and routing qubits required for each patch, and plot these against the circuit distance $d_{\mathrm{circuit}}$. 

For the various boundary geometries indicated, at the highest verified circuit distance, we already see that these codes are more efficient than $k$ equivalent patches of the rotated surface code with the same circuit distance. \cref{tab:performing_boundary_cuts} shows $n_{\mathrm{total}} / n_{\mathrm{surf\ total}}$ for selected vine sequences and boundary geometries. For the vine codes and patches shown in \cref{fig:nSeWwSs_results_fig2} and \cref{fig:nSeWWwSs_results_fig3}, the vine code uses fewer total qubits than the equivalent surface code construction. For example, for a rotated patch of the $\overrightarrow{nSeWwSs}$ code at distance $7$, the vine code requires $88\%$ of the total qubit count required for equivalent surface code patches. The table also shows the number of data and measure qubits ($n_{\mathrm{data+meas}}$, i.e. omitting routing qubits) required for the vine code relative to the surface code. While we use several techniques to reduce the total number of routing qubits (see \appref{app:swap_heuristics}), we envisage these reductions can be further improved.

We observe these vine codes have favourable scaling with respect to equivalent surface code patches. In \cref{fig:nSeWwSs_results_fig2} and \cref{fig:nSeWWwSs_results_fig3} we fit the qubit numbers $n_{\mathrm{total}}$ and $n$ to a polynomial ansatz with a quadratic and linear term. We expect a quadratic term similar to the surface code qubit scalings  $n_{\mathrm{surf}} = k d^2$ and $n_{\mathrm{surf\ total}} = k (2d^2 - 1)$  (where $n_{\mathrm{surf}}$ is the number of data qubits needed for $k$ rotated surface code patches and $n_{\mathrm{surf\ total}}$ is the total data and measure qubits needed for $k$ rotated surface code patches). We also expect the number of routing qubits required to be linear in $d$. The asymptopic scaling of the vine code overhead improvement in relation to the surface code $n_{\mathrm{total}} / n_{\mathrm{surf\ total}}$ is given by $C / 2k$.

We find that these vine codes have increasing returns on qubit numbers relative to the surface code as the circuit distance increases. For example, the $\overrightarrow{nSeWwSs}$ code on a rotated patch gives $C / 2k = 0.797$. In \cref{fig:code_patches} and \cref{tab:distance_10_codes}, we corroborate this claim by constructing $d=10$ patches for four boundary geometries (stim circuits for these code patches can be found at Ref.~\onlinecite{nixon_2026_20734084}). The rotated geometry requires $82\%$ of the total qubits required for equivalent surface code patches, a larger improvement than the $d = 7$ patch. We expect the circuit distance of the patches shown in \cref{fig:code_patches} to match the code distances, i.e. $d_{\mathrm{circuit}}=10$, as \cref{tab:performing_boundary_cuts} illustrates that the $\overrightarrow{nSeWwSs}$ code on all boundary geometries explored has no distance-reducing hook errors, given $d = d_{\mathrm{circuit}}$ in each case.

In \cref{tab:performing_boundary_cuts}, we give verified results for all codes where $n_{\mathrm{meas+data}}/n_{\mathrm{surf\ total}} < 1$ i.e. codes that at circuit distance $6$ already use fewer measure and data qubits than the surface code. For these examples, we further present the polynomial fitting parameters in \cref{fig:distance_scaling_appendix} in \appref{app:exhaustive_sequence_search}, finding $11$ examples where $C/2k <1$, suggesting improvements in qubit count over the surface code as the code distance increases. 

Our results indicate that the boundary cut geometry matters significantly for qubit counts. In \cref{tab:performing_boundary_cuts}, we show qubit numbers required for the same vine code on different boundary geometries, which give varying $n$ and $n_{\mathrm{total}}$ to achieve the same circuit distance. We also see that different cuts give varying biases in the length of $v_1$ and $v_2$ to achieve equal circuit distances in both bases. While this does not affect the minimum weight logical size by construction, this may affect the relative number of minimum weight logical representatives between bases. This entropic effect may be seen in circuit-level noise simulations, see \cref{sec:circuit-level_noise_simulations}.

\subsection{Circuit-Level Noise Simulations}
\label{sec:circuit-level_noise_simulations}
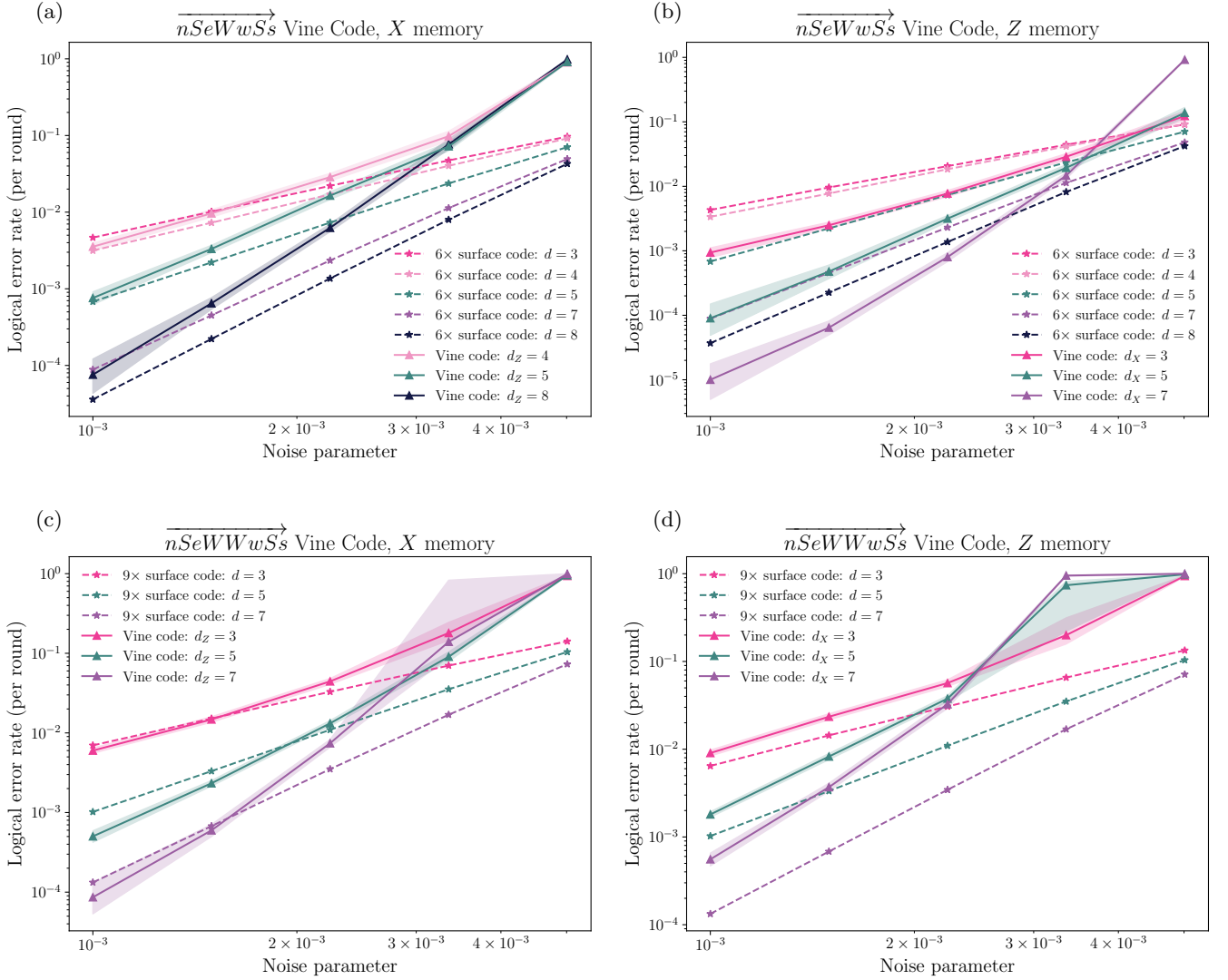
\begin{figure*}
    \centering
    \begin{tikzpicture}
  \node[anchor=south west, inner sep=0] (img) at (0,7.5)
    {\includegraphics[height=0.39\linewidth]{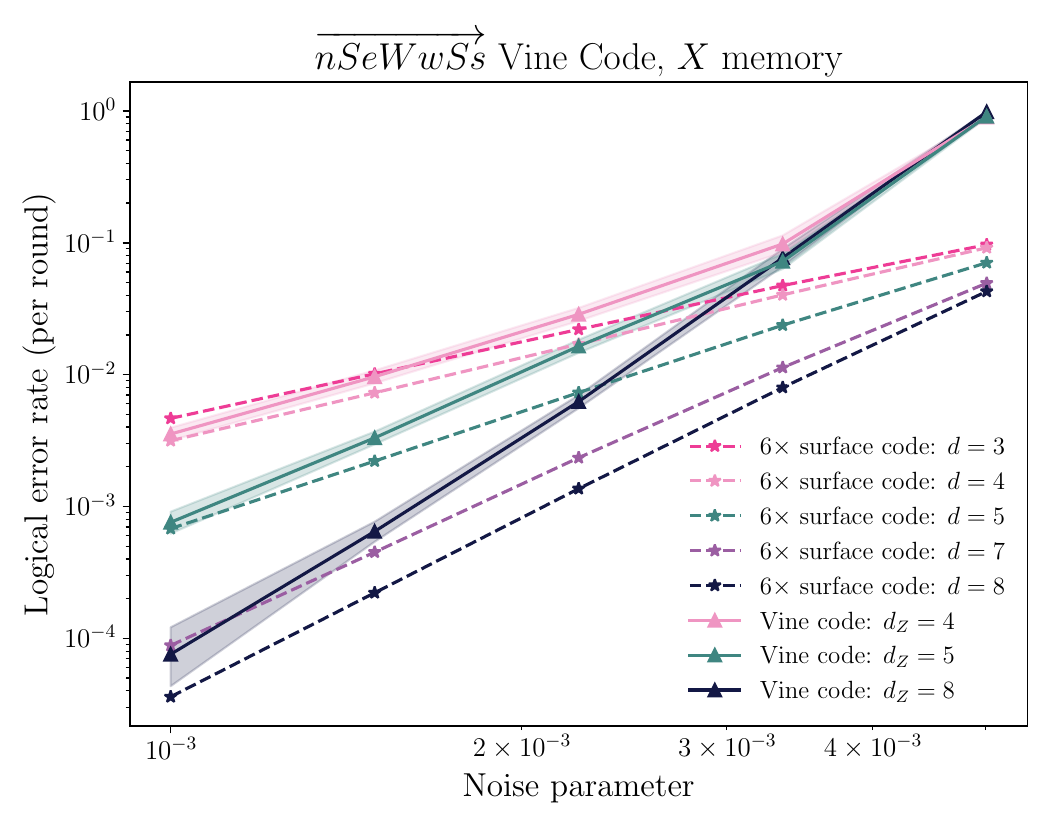}};
\node[anchor=south west, inner sep=0] (img) at (9,7.5)
    {\includegraphics[height=0.39\linewidth]{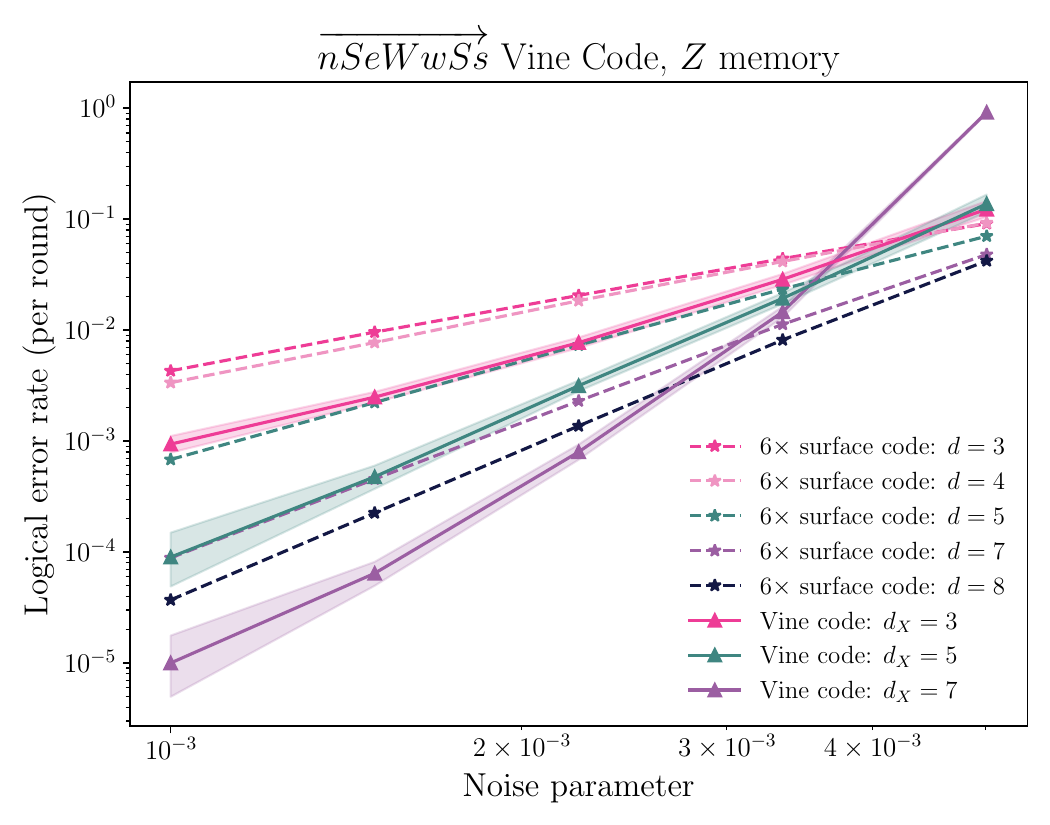}};
\node[anchor=south west, inner sep=0] (img) at (0,0)
    {\includegraphics[height=0.39\linewidth]{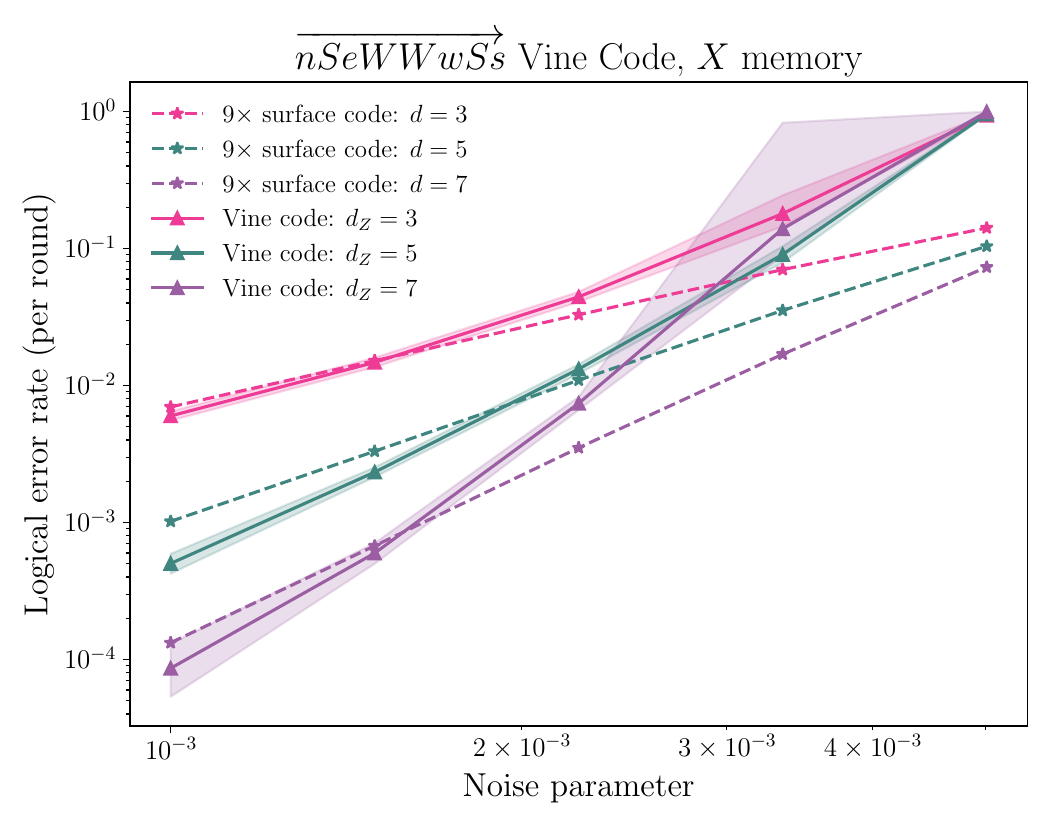}};
\node[anchor=south west, inner sep=0] (img) at (9,0)
    {\includegraphics[height=0.39\linewidth]{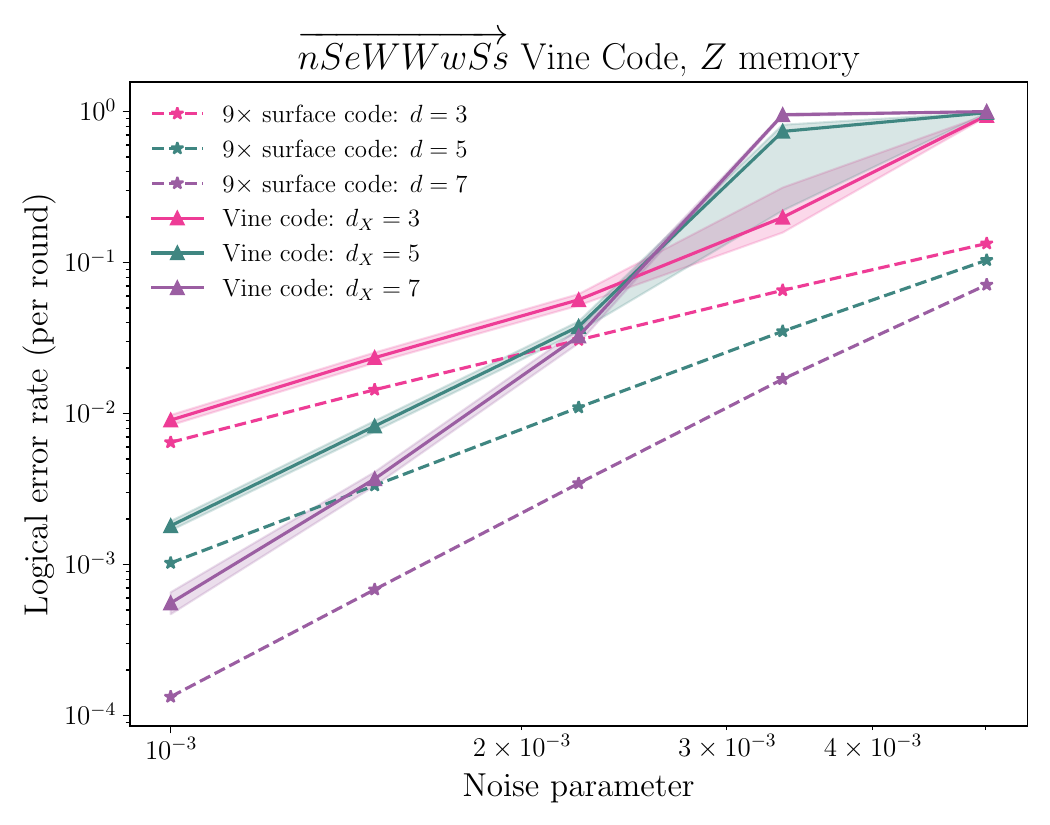}};
    \node[anchor=north west] at (0.5, 14.5) {(a)};
    \node[anchor=north west] at (9.5, 14.5) {(b)};
    \node[anchor=north west] at (0.5, 7.1) {(c)};
    \node[anchor=north west] at (9.5, 7.1) {(d)};
\end{tikzpicture}
\caption{Comparing the performance of a vine code quantum memory circuit encoding $k$ qubits with the performance of $k$ copies of the surface code. We use an SI1000 noise model for both types of circuit~\cite{gidney2023yokedsurfacecodes}, a BPOSD decoder for vine codes~\cite{Roffe_LDPC_Python_tools_2022} and PyMatching for the surface code~\cite{PyMatching2023}. \textbf{(a)}-\textbf{(b)} Memory experiments in $X$ and $Z$ bases are plotted for code patches of the $\overrightarrow{nSeWwSs}$ code. The patches have rotated rectangular vectors $v_1 = [-s,s]$, $v_2 = [2s,2s]$ for $s=5,7,11$, and circuit distances $(d_X, d_Z)$ of $(3,4)$, $(5,5)$ and $(7,8)$, respectively. \textbf{(c)}-\textbf{(d)} Memory experiments are plotted for patches of the $\overrightarrow{nSeWWwSs}$ code, with horizontal parallelogram vectors $v_1 = [s,s]$, $v_2=[s+2,0]$ for $s=10,18,26$. The vectors were chosen in both cases so that the patches had minimum circuit distances 3, 5, and 7. We report the circuit-distance that affects the relevant memory experiment ($d_Z$ for $X$ memory, $d_X$ for $Z$ memory) in each subfigure, though note that the code patches simulated are the same between $X$ and $Z$ memory experiments.\label{fig:sim_results}}
\end{figure*}

We perform numeric benchmarks of some of the vine codes found, and present our results in Fig.~\ref{fig:sim_results}. We find that under realistic assumptions about near-term hardware, vine codes can perform better than the surface code while lowering overheads. This can be seen, for example, in Fig.~\ref{fig:sim_results}(b)-(c).

Our circuits use only hardware-native gates (up to single-qubit Cliffords, which are typically very low noise in real devices) and a realistic, SI1000 noise model for superconducting qubits~\cite{gidney2023yokedsurfacecodes}. We choose patches to simulate with circuit distances $3$, $5$, and $7$, and compare these to surface codes of equivalent distances, in both $X$ and $Z$ simulated quantum memory experiments. See \appref{app:numerics_details} for more details of our numerical simulations.

The $Z$ memory performance of the $\overrightarrow{nSeWwSs}$ vine code is better by almost an order of magnitude than the corresponding surface code patch of the same distance, at a noise rate of $p=10^{-3}$, whereas for $X$ memory the vine code is roughly equivalent, compared to a distance-7 surface code. This is due to some asymmetry in the logical operators of the vine code. For example, the $(d_X=3,d_Z=4)$ patch has only a single logical qubit that has a weight-3 $X$-logical operator, while the other independent $X$ logical strings have weights 4 (for three strings) or 6 (for two strings). Meanwhile, there are six weight-4, independent $Z$-logical strings. Additionally, the asymmetric shape of the code patch provides different entropic factors for $X$ and $Z$ memory experiments. Although some patches simulated here have slightly uneven $d_X$ and $d_Z$, given that our resource comparisons in Section~\ref{sec:qubit_overhead_scaling} compare vine codes to their corresponding $d\times d$ surface code patch, we perform that same comparison here, reporting the results for $d\times d$ surface code memory performance (with SI1000 noise). We scale the per-round logical error rate of a single surface code patch (via the formula $P_k = 1-(1-P_L)^k$) to estimate the logical error rate of $k$ logical qubits encoded in surface codes.  

\subsection{Flip-vine codes}
\label{sec:flip-vine_code_results}

\begin{table}
    \centering
    \renewcommand{\arraystretch}{1.5}
    \begin{tabular}{ccccc}
    \toprule
    \textbf{Vine sequence} & \textbf{Code} & $v_1$ & $v_2$ & $\frac{n_\text{flip}}{n_\text{toric}}$\\
    \midrule
    $\overrightarrow{neSWWSws}$ & $  [[144,16,6]]$ & $[12, 0]$ & $[0, 24]$ & 0.25\\
    \midrule
    $\overrightarrow{nSSEESSs}$  & $[[120,6,10]]$ & $[12,0]$ & $[0,20]$ & 0.2\\
    \midrule
    $\overrightarrow{nSENNESs}$ & $[[72,10,6]]$ & $[12,0]$ & $[0,12]$ & 0.2\\
    \midrule
    $\overrightarrow{nSsEEnSs}$ & $[[24, 8, 4]]$ & $[8, 0]$ & $[6,6]$ & 0.189\\
    \midrule
    $\overrightarrow{nSeWWwSs}$ & $[[32, 12, 4]]$ & $[8,0]$ & $[0,8]$ & 0.167\\
    \midrule
    $\overrightarrow{nSEEEESs}$ & $[[72, 12, 6]]$ & $[12, 0]$ & $[0,12]$ & 0.167
    \end{tabular}
    \caption{Examples of flip-vine codes with periodic boundary conditions that permit transversal single-qubit Clifford gates. We report the step sequence, the code parameters, the vectors $v_1$ and $v_2$ defining the torus, along with the code's data qubit efficiency relative to the toric code, $n_\text{flip}/n_\text{toric}$, where $n_\text{flip}$ ($n_\text{toric}$) is the number of data qubits required for the flip-vine (toric) code encoding of the same number of logical qubits at the same distance. Apart from the first row, we report here all codes found with $n_\text{flip}/n_\text{toric}\leq 0.2$, where we remove duplicates with the same code parameters.}
    \label{tab:flip_vine_table}
\end{table}

We provide here some examples of flip-vine codes found on a grid with periodic boundary conditions. We perform a search through the found valid step sequences up to and including length 8, and build example codes on several tori. These tori are defined with vectors $v_1 = [v,0]$, $v_2=[p,w]$ such that opposite boundaries of the parallelogram defined by these vectors are identified. We search through tori with $v$ and $w$ ranging between $4$ and $12$ and $p$ ranging between $0$ and $v$, and through rectangular tori ($p=0$) with $v=4,\ldots, 12$, $w=4,\ldots, 24$. For each size of code, we check that the wrapping of the infinite plane around the torus sends data qubits to data qubits, and $Z$/$X$ measure qubits to measure qubits of the same type. We ensure that the resulting code has commuting stabilisers. We then check the number of logical qubits and the distance (the $X$ and $Z$ distances are equivalent). In Table~\ref{tab:flip_vine_table}, we also supply the ratio of the number of data qubits in the flip-vine code and the number of data qubits required to encode the same $k$ logical qubits in toric codes.

\section{Conclusion}\label{sec:conclusion}


In this paper, we have introduced vine codes, quantum low-density parity check codes that can significantly reduce the overheads associated with quantum error-correction on a planar square grid. Our work is particularly relevant to superconducting hardware, where the connections between qubits are often restricted to $\leq 4$ nearest neighbours in a plane. Even taking into account the qubits required for routing around the boundary, there exist vine codes encoding $k\geq 6$ logical qubits that require fewer physical qubits than $k$ surface code patches of the same circuit-distance, and which perform equal to or better than the surface code under circuit-level noise. Compare, for example, the $\overrightarrow{nSeWwSs}$ codes in Fig~\ref{fig:sim_results}(a)-(b) to the larger-overhead distance-$3$, $5$ and $7$ surface code patches, at $p\approx 10^{-3}$. The gains in overheads are expected to improve as one goes to higher distance. While a two-dimensional translationally invariant code is unable to perform better asymptotically than having $k=\Theta(1)$ and $d=\Theta(\sqrt{n})$, one can nevertheless save considerable resources at the utility scale (up to $\lesssim 20.3\%$ based on our fits) relative to the surface code without requiring significantly more complex hardware. 

There are many directions still open for further research. We have not investigated in depth the possibilities of generalised boundaries or patches of different shapes (e.g., different parallelograms, triangles, pentagons, etc.). Future work could investigate more exotic boundaries, such as Pauli-$Y$ boundaries (possible, for example, in the 2D color code~\cite{Kesselring_2018_bdrys_twists}), or further generalisations of the coloured boundaries discussed in \cref{sec:gen_bdrys}. Future work could also consider general prescriptions for constructing syndrome extraction circuits for generalised boundaries.

Surface codes have benefited from decades of research, and there exist many techniques that allow for further reductions in overheads or improvements to logical fidelity. These include yoking~\cite{gidney2023yokedsurfacecodes,low2026denserplanarsurfacecode}, leakage-reduction~\cite{Battistel_2021,McEwen_2023_relaxing}, tailoring to biased noise~\cite{BonillaAtaides2021XZZX,Tuckett2019Tailoring}, mitigation of defective components~\cite{Auger_2017,Debroy_2025,wolanski2026automatedcompilationincludingdropouts}, and the use of fast and accurate decoders~\cite{PyMatching2023,Delfosse_2021_union_find,shutty2024efficientnearoptimaldecodingsurface}. Such improvements could be brought to bear on vine codes as well. In particular, it is important to investigate the performance of decoders with improved speeds~\cite{relay_bp,Hillmann_2025,wolanski2025ambiguityclusteringaccurateefficient,koutsioumpas2025colourcodesreachsurface,sahay2026matchingdecoderbivariatebicycle,tan2026generalizedmatchingdecoders2d}. One large avenue for future work is to understand whether it is possible to adapt existing techniques for fault-tolerant logic, such as lattice surgery, magic state cultivation and distillation, to integrate vine codes into a fault-tolerant architecture. Achieving arbitrary Clifford gates or multi-Pauli product measurements between the $k$ logical qubits stored within a single vine code patch (necessary for universal quantum computation) may be challenging without re-introducing large resource overheads or long-range connectivity requirements. An alternative to this could be to construct vine code patches with more exotic boundary conditions storing single logical qubits. These are important directions for future work. Adapting flip-vine codes for magic state cultivation and investigating their circuit-level performance could provide an additional fruitful avenue for further research.

Our work has assumed a square grid layout, but it may be interesting for future work to investigate the tradeoffs associated with scaling the connectivity. For example, the $\overrightarrow{sEEEn}$ (or equivalently, the $\overrightarrow{NEEEN}$) step sequence only requires a hex-grid connectivity~\cite{Geher2025Directional}. Meanwhile, other codes may exist that could achieve higher performance if one allows for a tiling in which all vertices are five- or six-valent. From a theoretical standpoint, it would be fascinating to consider the possibilities of vine codes when embedded on a hyperbolic plane, in which one can expect asymptotically constant-rate and logarithmic distance codes~\cite{Breuckmann_2016_hyperbol,Breuckmann_2017_hyperbol_2,Higgott_2024_Hyperbol_Floquet,Fahimniya_2025}. Meanwhile, generalising the dynamic sequence introduced in the flip-vine code construction to construct novel Floquet codes~\cite{Hastings_2021,Haah_2022_bdrys,Gidney_2021,Gidney_2022,vuillot2021planarfloquetcodes,Davydova_2023,Kesselring_2024_anyon_condensation,Higgott_2024_Hyperbol_Floquet,Fahimniya_2025,Setiawan_2025,McLauchlan_2024,jacoby2026stairwaycodesfloquetifyingbivariate} on a planar square grid could similarly produce interesting resource savings and stabiliser weight reductions.

While we have counted the routing qubits towards the total cost of implementing a vine code patch, it is unclear how fair of an assumption this is. Indeed, surface codes also require a buffer layer of qubits around the boundary if one is to perform arbitrary lattice surgery operations with neighbouring patches~\cite{low2026denserplanarsurfacecode}. It is possible that the routing qubits we have counted in this work could be absorbed into a ``corridor" of qubits that would already be required for lattice surgery. 

A challenge for implementing vine codes on today's devices involves calibrating superconducting hardware for both CZ and iSWAP gates and achieving the high gate fidelities seen in recent years for both gates simultaneously. That being said, many of the vine codes we have investigated do not strictly require both gates to implement. At the cost of potentially introducing more routing qubits, sequences that contain no lower-case letters in their middle can easily be converted to a fully-upper-case sequence which is implementable solely with iSWAP gates. See, for example, some of the sequences in Table~\ref{tab:performing_boundary_cuts}. However, experiments have demonstrated that it is possible to function superconducting devices with iSWAP and CZ gates operating together~\cite{iswap_with_cz}. The field has focused extensively on improving CZ gate performance in particular, and hence there may still be large unrealised gains to be found in the fidelities of iSWAP gates operated together with CZ gates.

\section{Acknowledgments}
We are grateful to Timo Hillmann, Aislin Wells, György Geher, David Byfield, Archibald Ruban, David Long, Andrew Doherty and Stephen Bartlett for helpful and encouraging discussions. 

During the preparation of this manuscript, it was brought to our attention that similar problems were being worked on concurrently by another group. This concurrent work is available at Ref.~\onlinecite{Boren_nearest_neighbour_gates}. We thank Boren Gu and the other authors of this manuscript for coordinating with us to release our works jointly.

GMN and CKM acknowledge support from the Intelligence Advanced Research Projects Activity (IARPA), under the Entangled Logical Qubits program through Cooperative Agreement Number W911NF-23-2-0223.
CCLVR acknowledges support from the University of Sydney Nano Institute through the Taste of Research Award.

The python code produced for this manuscript was written with the assistance of large language models. The ideas and writing are the authors' own.

\appendix

\section{Valid sequences}
\label{app:valid_sequences}
Not all vine code sequences are valid codes. Here, we discuss the four conditions for a valid code: commutativity, independent measurement, topological order, and patch connectivity. 

In \appref{app:exhaustive_sequence_search}, we give all the vine codes (that include at least one CXSWAP gate) up to stabiliser weight $9$ that satisfy the four criteria detailed here. We also ignore sequences with steps that entangle then immediately unentangle, for example any sequence with $\vec{E}$ directly followed by $\vec{w}$, or $\vec{N}$ followed by $\vec{s}$.  We assume the patch has alternating rows of X- and Z-stabilisers, corresponding to layout 1 in Ref.~\cite{Geher2025Directional}, (see also \cref{fig:vine_codes}).

\textbf{1. Commutativity:}
Stabilisers of opposite type must have even overlapping support. 

\textbf{2. Independent measurement:}
 Measure qubits must not be  entangled or ``interleaved" at the measurement step. This requires that, for $X$- and $Z$-stabilisers with common data support, the subset of data qubits entangled with the $X$-measure qubit before the $Z$-measure qubit must be of even size, and vice versa. 

Let us define the set of data qubits $\mathcal{D} = \lbrace q_1, q_2, \ldots, q_n\rbrace$. We then define, for measure qubit $a$, $\mathcal{Q}_a = \lbrace Q_a^i \, | \, i=1,\ldots, w\rbrace$, where $w$ is the length of the step sequence, such that $Q_a^i\in \mathcal{D}$ is the $i^{\text{th}}$ data qubit entangled with measure qubit $a$. A step sequence and layout then defines a code with independently measured stabilisers  if~\cite{Geher2025Directional,Geher_tangling_schedules_2024} the set $\lbrace (i,j)\, |\, Q_a^i = Q_b^j ,\, j > i, \, a\neq b\rbrace$ has even cardinality whenever measure qubits $a$ and $b$ are of different Pauli types. 

Note that we consider vine codes that satisfy this condition on the layout with alternating rows of $X$ and $Z$ stabilisers as in \cref{fig:vine_codes} in our exhaustive search in \appref{app:exhaustive_sequence_search}. However, some vine codes may fail this condition on this layout but succeed on other layouts, for example, $\overrightarrow{NENWSW}$. Theorem 1 of Ref.~\onlinecite{Geher2025Directional} can be used to determine which layouts, if they exist, will satisfy this condition for a given step sequence (see also \appref{app:delta_sets}). It is straightforward to confirm that this Theorem holds for the more general vine code sequences in our paper, despite the inclusion of non-swapping entangling gates.

\textbf{3. Topological order condition:}
Some sequences that satisfy conditions \textbf{1} and \textbf{2} above do not produce topologically ordered codes. In this case, all their codewords are a fixed weight and do not increase with the system size. For example, $\overrightarrow{NENNEN}$, $\overrightarrow{NENNNEN}$, $\overrightarrow{NEENNEEN}$, $\overrightarrow{NENNENNEN}$ and $\overrightarrow{NNENNNENN}$ are codes that have been previously referenced as valid codes \cite{Geher2025Directional, rowshan2026directional}, but fail to have topological order. In some cases, the number of logical qubits also grows with the system size, but these codewords remain fixed and low weight. 

Some codes have a subset of logicals that have topological order, and a subset that remain fixed weight. For example $\overrightarrow{neSWWSws}$, $\overrightarrow{nEENNNEEs}$, $\overrightarrow{nEENENEEs}$, $\overrightarrow{nEENWNEEs}$ and $\overrightarrow{nEEnSsEEs}$. The fixed weight logicals in these cases may be an artifact of the corner boundary conditions we impose in our construction, see \appref{app:corners}, and should be treated with caution. However, we include these examples in \cref{tab:topological_vine_codes_wt1-9}, as some logicals are topological and we may obtain a high-distance code after including the low-weight logicals in the stabiliser group.

\textbf{4. Patch connectivity:} Some vine codes create multiple, independent and disconnected code regions, see \cref{fig:bad_patch_connectivity}. In our exhaustive search, we build vine codes using a square patch with boundaries $v_1 = 30$, and $v_2 = 30$, which initially contains $450$ data qubits. During the boundary formation, where stabilisers are truncated and removed, data qubits that are not supported on a stabiliser of both types around the boundary are removed. While we expect this boundary formation procedure may remove some boundary data qubits, giving a final data qubit number below $450$, a code that creates multiple disconnected code regions will reduce the final data count by approximately half, or even fewer if more than $2$ disconnected regions are created. Therefore, in \cref{tab:topological_vine_codes_wt1-9} we omit codes that resulted in less than $250$ final data qubits (some codes that divide the patch similarly to \cref{fig:bad_patch_connectivity} have just over half the total data qubits in one disconnected region, for example $\overrightarrow{nEewSewEs}$ has $227$ qubits in the largest disconnected region).

\begin{figure}
    \centering
    \includegraphics[width=\linewidth]{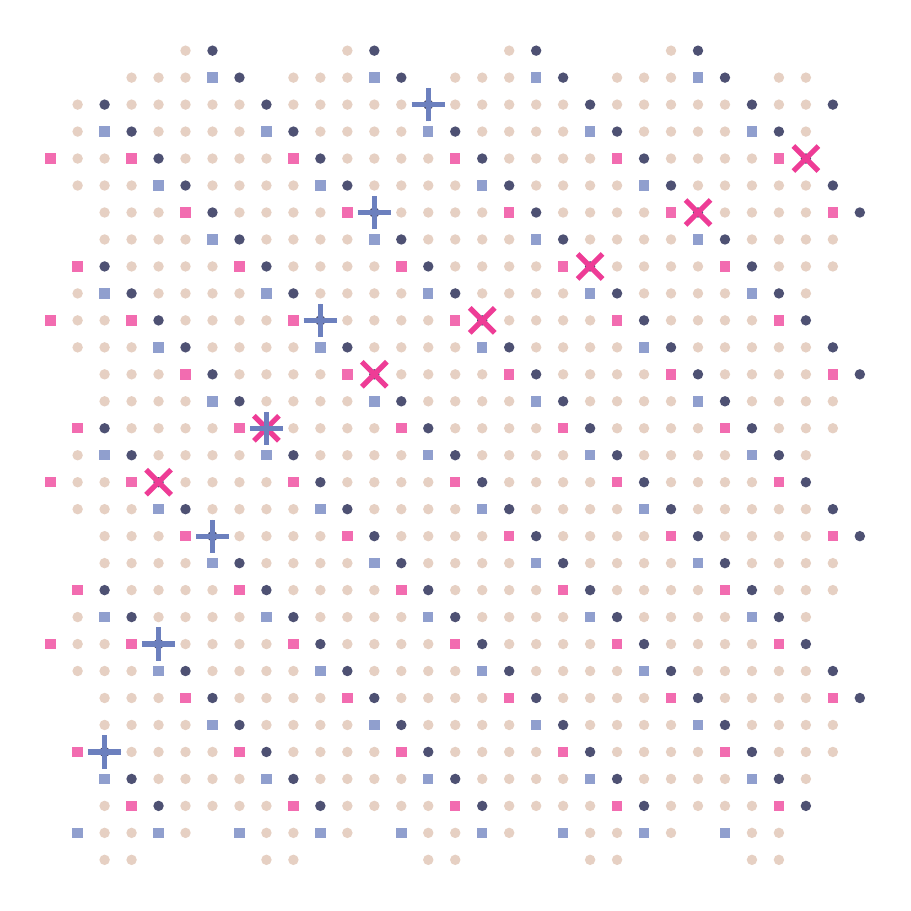}
    \caption{The $\overrightarrow{nENnsNEs}$ code divides the patch into three, approximately equal disconnected regions. Here we show one region, where data qubits are black, $Z$ measure qubits are blue, $X$ measure qubits are pink, and routing qubits are brown. The data and measure qubits of the other two regions are supported on the edge qubits. The support of the $X$ ($Z$) logical operators are shown by the pink crosses (blue crosses). This patch encodes three disconnected logical qubits.}
    \label{fig:bad_patch_connectivity}
\end{figure}

\section{Exhaustive vine code sequence search}
\label{app:exhaustive_sequence_search}

We perform exhaustive searches over all vine sequences up to length $9$. We remove sequences that do not satisfy the four conditions of \appref{app:valid_sequences}.

We also remove sequences based on the following additional criteria:

\textbf{5. $D_4$ sequence symmetry:} We consider vine sequences modulo $D_4$ (reflection and rotation by $\pi / 2$ symmetry). For example, the vine code $\overrightarrow{nEEs}$ is a reflection of the sequence $\overrightarrow{sEEn}$. 

\textbf{6. Begin and end with CX gate:} We remove sequences that begin or end with a CXSWAP gate, i.e. the instruction $N$, $E$, $S$ or $W$. This is because there exists an equivalent vine sequence that begins and ends with a CX gate that has identical data qubit support and entangling order. Consider replacing the CXSWAP that begins a sequence with the CX gate in the opposite cardinal direction, i.e. $\vec{N} \rightarrow \vec{s}$, $\vec{E} \rightarrow \vec{w}$, $ \vec{S} \rightarrow \vec{n}$ and $ \vec{W} \rightarrow \vec{e}$. Also consider replacing the CXSWAP gate that ends a sequence with the CX gate in the same cardinal direction, i.e. $\vec{N} \rightarrow \vec{n}$, $\vec{E} \rightarrow \vec{e}$, $ \vec{S} \rightarrow \vec{s}$ and $ \vec{W} \rightarrow \vec{w}$. These transformations will form an equivalent code, where the first and last swap is avoided. Avoiding fruitless swaps can reduce the number of routing qubits required. For example $\overrightarrow{NEEN}$ is equivalent to $\overrightarrow{sEEn}$.

\textbf{7. $D_4$ support symmetry:} We find codes that generate equivalent data qubit support modulo $D_4$ while entangling with data qubits in a different order. For example, $\overrightarrow{nSsEEnSs}$ and $\overrightarrow{nsENNEsn}$. Filtering sequences with equivalent support symmetry modulo $D_4$ reduces the set of vine codes we generate from $906$ to $147$. We remove sequences with identical support and include only one representative in \cref{tab:topological_vine_codes_wt1-9}. We include all representatives before this reduction in  Ref.~\onlinecite{nixon_2026_20734084}.

Satisfying conditions \textbf{1}-\textbf{4} in \appref{app:valid_sequences} and reducing based on criteria \textbf{5}-\textbf{7} above, we arrive at  $146$ valid vine codes. These are given in \cref{tab:topological_vine_codes_wt1-9}. After reducing based on criteria \textbf{5}-\textbf{7}, we include only one canonical representative for each set of codes with equivalent stabiliser support (prioritising cardinal directions in their clockwise order in the step sequences, beginning with $N$ or $n$).

Importantly, the reduction in step \textbf{7} reduces the number of remaining codes dramatically. Without this reduction, while still reducing based on criteria \textbf{1}-\textbf{6}, we would arrive at $906$ codes. The representative of each code in \cref{tab:topological_vine_codes_wt1-9} may have different hook errors to a vine code that generates equivalent support via a different step sequence. We do not optimise for hook errors when choosing a representative, since the effects of hook errors depend on the final code patch one chooses. This point may be important to explore when further optimising vine codes for performance on near term architectures and is left for future work.


\begin{table*}[!ht]
\centering
\renewcommand{\arraystretch}{1.25}
\begin{tabular}{@{}l@{\;}l@{\hspace{3em}}l@{\;}l@{\hspace{3em}}l@{\;}l@{\hspace{3em}}l@{\;}l@{}}
\multicolumn{8}{c}{\textbf{Weight 4}} \\
\noalign{\vskip 0.15em}
\hline
\noalign{\vskip 0.1em}
$\overrightarrow{nEEs}$ & $(k=1)$ &  &  &  &  &  &  \\
\noalign{\vskip 0.25em}
\hline
\noalign{\vskip 0.45em}
\multicolumn{8}{c}{\textbf{Weight 5}} \\
\noalign{\vskip 0.15em}
\hline
\noalign{\vskip 0.1em}
$\overrightarrow{nEEEs}$ & $(k=2)$ & $\overrightarrow{nENEs}$ & $(k=2)$ &  &  &  &  \\
\noalign{\vskip 0.25em}
\hline
\noalign{\vskip 0.45em}
\multicolumn{8}{c}{\textbf{Weight 6}} \\
\noalign{\vskip 0.15em}
\hline
\noalign{\vskip 0.1em}
$\overrightarrow{neSSws}$ & $(k=3)$ & $\overrightarrow{nEEEEs}$ & $(k=3)$ & $\overrightarrow{nENNEs}$ & $(k=3)$ & $\overrightarrow{nSewSs}$ & $(k=3)$ \\
$\overrightarrow{nSEESs}$ & $(k=3)$ &  &  &  &  &  &  \\
\noalign{\vskip 0.25em}
\hline
\noalign{\vskip 0.45em}
\multicolumn{8}{c}{\textbf{Weight 7}} \\
\noalign{\vskip 0.15em}
\hline
\noalign{\vskip 0.1em}
$\overrightarrow{nEnEsEs}$ & $(k=2)$ & $\overrightarrow{nEnSsEs}$ & $(k=2)$ & $\overrightarrow{nEESEEs}$ & $(k=2)$ & $\overrightarrow{nENNWSe}$ & $(k=2)$ \\
$\overrightarrow{nESWSEs}$ & $(k=2)$ & $\overrightarrow{nSESESs}$ & $(k=2)$ & $\overrightarrow{neSSSws}$ & $(k=4)$ & $\overrightarrow{nEeSwEs}$ & $(k=4)$ \\
$\overrightarrow{nEEEEEs}$ & $(k=4)$ & $\overrightarrow{nEENEEs}$ & $(k=4)$ & $\overrightarrow{nENENEs}$ & $(k=4)$ & $\overrightarrow{nENNNEs}$ & $(k=4)$ \\
$\overrightarrow{nSeSwSs}$ & $(k=4)$ & $\overrightarrow{nSENESs}$ & $(k=4)$ & $\overrightarrow{nSeWwSs}$ & $(k=6)$ & $\overrightarrow{nSEEESs}$ & $(k=6)$ \\
\noalign{\vskip 0.25em}
\hline
\noalign{\vskip 0.45em}
\multicolumn{8}{c}{\textbf{Weight 8}} \\
\noalign{\vskip 0.15em}
\hline
\noalign{\vskip 0.1em}
$\overrightarrow{newSSews}$ & $(k=1)$ & $\overrightarrow{nsEnsEns}$ & $(k=1)$ & $\overrightarrow{nEsNNnEs}$ & $(k=1)$ & $\overrightarrow{nENWNWSe}$ & $(k=1)$ \\
$\overrightarrow{nENWSene}$ & $(k=1)$ & $\overrightarrow{nESsnSEs}$ & $(k=1)$ & $\overrightarrow{nESSSSEs}$ & $(k=1)$ & $\overrightarrow{nSESSESs}$ & $(k=1)$ \\
$\overrightarrow{nEnSsnEs}$ & $(k=2)$ & $\overrightarrow{nsENNEns}$ & $(k=3)$ & $\overrightarrow{nEnSSsEs}$ & $(k=3)$ & $\overrightarrow{nSsEEnSs}$ & $(k=3)$ \\
$\overrightarrow{nSEnsESs}$ & $(k=3)$ & $\overrightarrow{neSSSSws}$ & $(k=5)$ & $\overrightarrow{neSWWSws}$ & $(k=5)$ & $\overrightarrow{nEEEEEEs}$ & $(k=5)$ \\
$\overrightarrow{nENEENEs}$ & $(k=5)$ & $\overrightarrow{nENNNNEs}$ & $(k=5)$ & $\overrightarrow{nESWWSEs}$ & $(k=5)$ & $\overrightarrow{nSeSSwSs}$ & $(k=5)$ \\
$\overrightarrow{nSENNESs}$ & $(k=5)$ & $\overrightarrow{nSSewSSs}$ & $(k=5)$ & $\overrightarrow{nSSEESSs}$ & $(k=5)$ & $\overrightarrow{nSeWweSs}$ & $(k=6)$ \\
$\overrightarrow{nEeSSwEs}$ & $(k=7)$ & $\overrightarrow{nSeWWwSs}$ & $(k=9)$ & $\overrightarrow{nSEEEESs}$ & $(k=9)$ &  &  \\
\noalign{\vskip 0.25em}
\hline
\noalign{\vskip 0.45em}
\multicolumn{8}{c}{\textbf{Weight 9}} \\
\noalign{\vskip 0.15em}
\hline
\noalign{\vskip 0.1em}
$\overrightarrow{newSESews}$ & $(k=2)$ & $\overrightarrow{nsNENENns}$ & $(k=2)$ & $\overrightarrow{nEneSwsEs}$ & $(k=2)$ & $\overrightarrow{nEnESEsEs}$ & $(k=2)$ \\
$\overrightarrow{nEnSWSsEs}$ & $(k=2)$ & $\overrightarrow{nEsNNNnEs}$ & $(k=2)$ & $\overrightarrow{nEEnWsEEs}$ & $(k=2)$ & $\overrightarrow{nEEsNnEEs}$ & $(k=2)$ \\
$\overrightarrow{nEEsWnEEs}$ & $(k=2)$ & $\overrightarrow{nEEESEewe}$ & $(k=2)$ & $\overrightarrow{nEESEEnsn}$ & $(k=2)$ & $\overrightarrow{nEESESEEs}$ & $(k=2)$ \\
$\overrightarrow{nENESEewe}$ & $(k=2)$ & $\overrightarrow{nENNWSwew}$ & $(k=2)$ & $\overrightarrow{nESsNnSEs}$ & $(k=2)$ & $\overrightarrow{nESSSSSEs}$ & $(k=2)$ \\
$\overrightarrow{nSESWSESs}$ & $(k=2)$ & $\overrightarrow{nesWWWnws}$ & $(k=4)$ & $\overrightarrow{neSEEESws}$ & $(k=4)$ & $\overrightarrow{neSSSewes}$ & $(k=4)$ \\
$\overrightarrow{nsENNNEns}$ & $(k=4)$ & $\overrightarrow{nEnEEEsEs}$ & $(k=4)$ & $\overrightarrow{nEnENEsEs}$ & $(k=4)$ & $\overrightarrow{nEnSESsEs}$ & $(k=4)$ \\
$\overrightarrow{nEnSSSsEs}$ & $(k=4)$ & $\overrightarrow{nEeNENwEs}$ & $(k=4)$ & $\overrightarrow{nEsEEEnEs}$ & $(k=4)$ & $\overrightarrow{nEsNENnEs}$ & $(k=4)$ \\
$\overrightarrow{nEEnEsEEs}$ & $(k=4)$ & $\overrightarrow{nEEsEnEEs}$ & $(k=4)$ & $\overrightarrow{nEEEEEnsn}$ & $(k=4)$ & $\overrightarrow{nEEESEEEs}$ & $(k=4)$ \\
$\overrightarrow{nEENEEnsn}$ & $(k=4)$ & $\overrightarrow{nEESSSEEs}$ & $(k=4)$ & $\overrightarrow{nEESWSEEs}$ & $(k=4)$ & $\overrightarrow{nENnEsNEs}$ & $(k=4)$ \\
$\overrightarrow{nENeWwNEs}$ & $(k=4)$ & $\overrightarrow{nENENEnsn}$ & $(k=4)$ & $\overrightarrow{nENESENEs}$ & $(k=4)$ & $\overrightarrow{nENNNEnsn}$ & $(k=4)$ \\
$\overrightarrow{nENNWNNEs}$ & $(k=4)$ & $\overrightarrow{nESeWwSEs}$ & $(k=4)$ & $\overrightarrow{nESsEnSEs}$ & $(k=4)$ & $\overrightarrow{nESwNeSEs}$ & $(k=4)$ \\
$\overrightarrow{nESwSeSEs}$ & $(k=4)$ & $\overrightarrow{nESSWSSEs}$ & $(k=4)$ & $\overrightarrow{nESWNEESs}$ & $(k=4)$ & $\overrightarrow{nSewSewSs}$ & $(k=4)$ \\
$\overrightarrow{nSsENEnSs}$ & $(k=4)$ & $\overrightarrow{nSEeSwESs}$ & $(k=4)$ & $\overrightarrow{nSEsNnESs}$ & $(k=4)$ & $\overrightarrow{nSESESESs}$ & $(k=4)$ \\
$\overrightarrow{nSSeNwSSs}$ & $(k=4)$ & $\overrightarrow{nSSESESSs}$ & $(k=4)$ & $\overrightarrow{neSwEeSws}$ & $(k=6)$ & $\overrightarrow{neSSSSSws}$ & $(k=6)$ \\
$\overrightarrow{neSSWSSws}$ & $(k=6)$ & $\overrightarrow{nsEENEEns}$ & $(k=6)$ & $\overrightarrow{nEeNNNwEs}$ & $(k=6)$ & $\overrightarrow{nEEnSsEEs}$ & $(k=6)$ \\
$\overrightarrow{nEEeSwEEs}$ & $(k=6)$ & $\overrightarrow{nEEEEEEEs}$ & $(k=6)$ & $\overrightarrow{nEEENEEEs}$ & $(k=6)$ & $\overrightarrow{nEENWNEEs}$ & $(k=6)$ \\
$\overrightarrow{nENeNwNEs}$ & $(k=6)$ & $\overrightarrow{nENEEENEs}$ & $(k=6)$ & $\overrightarrow{nENENENEs}$ & $(k=6)$ & $\overrightarrow{nENNENNEs}$ & $(k=6)$ \\
$\overrightarrow{nENNNNNEs}$ & $(k=6)$ & $\overrightarrow{nESwEeSEs}$ & $(k=6)$ & $\overrightarrow{nESWSWSEs}$ & $(k=6)$ & $\overrightarrow{nSesWnwSs}$ & $(k=6)$ \\
$\overrightarrow{nSeSSSwSs}$ & $(k=6)$ & $\overrightarrow{nSsEEEnSs}$ & $(k=6)$ & $\overrightarrow{nSEnEsESs}$ & $(k=6)$ & $\overrightarrow{nSEeWwESs}$ & $(k=6)$ \\
$\overrightarrow{nSEsEnESs}$ & $(k=6)$ & $\overrightarrow{nSENNNESs}$ & $(k=6)$ & $\overrightarrow{nSSeSwSSs}$ & $(k=6)$ & $\overrightarrow{nSSENESSs}$ & $(k=6)$ \\
$\overrightarrow{neSWWWSws}$ & $(k=8)$ & $\overrightarrow{nEeSWSwEs}$ & $(k=8)$ & $\overrightarrow{nEENENEEs}$ & $(k=8)$ & $\overrightarrow{nESWWWSEs}$ & $(k=8)$ \\
$\overrightarrow{nSeSWSwSs}$ & $(k=8)$ & $\overrightarrow{nSeWNWwSs}$ & $(k=8)$ & $\overrightarrow{nSEESEESs}$ & $(k=8)$ & $\overrightarrow{nSENENESs}$ & $(k=8)$ \\
$\overrightarrow{nEeSSSwEs}$ & $(k=10)$ & $\overrightarrow{nEENNNEEs}$ & $(k=10)$ & $\overrightarrow{nSeWSWwSs}$ & $(k=10)$ & $\overrightarrow{nSEENEESs}$ & $(k=10)$ \\
$\overrightarrow{nSSeWwSSs}$ & $(k=10)$ & $\overrightarrow{nSSEEESSs}$ & $(k=10)$ & $\overrightarrow{nSeWWWwSs}$ & $(k=12)$ & $\overrightarrow{nSEEEEESs}$ & $(k=12)$ \\
\noalign{\vskip 0.25em}
\hline
\noalign{\vskip 0.45em}
\end{tabular}
\caption{Topologically ordered vine code sequences up to length 9. We find the above sequences from an exhaustive search of vine codes that satisfy conditions 1-4 in \appref{app:valid_sequences}, and filter for representatives via steps 5-7 in \appref{app:exhaustive_sequence_search}. This search assumes a layout with alternating rows of $X$- and $Z$-stabilisers, see \cref{fig:vine_codes}}
\label{tab:topological_vine_codes_wt1-9}
\end{table*}

\begin{figure*}
    \centering
    \includegraphics[width=\linewidth]{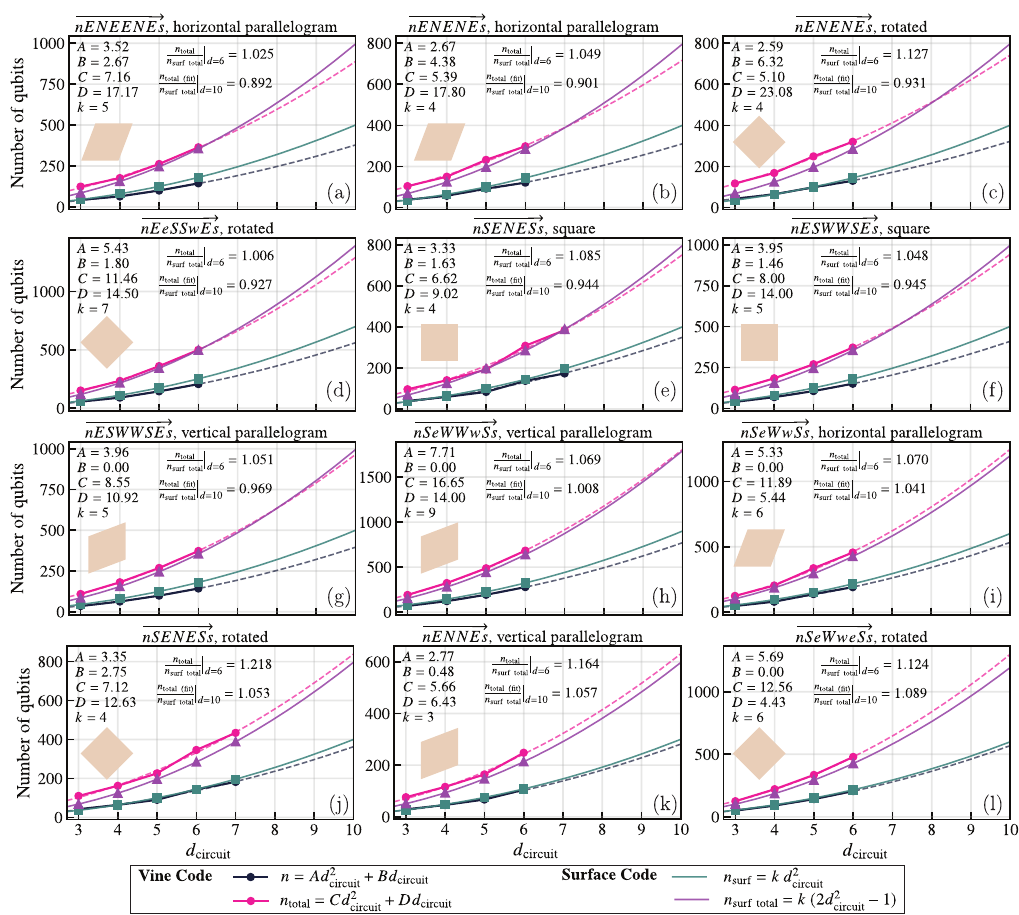}
    \caption{Qubit number versus circuit distance plots for all codes in \cref{tab:performing_boundary_cuts}. The conventions are equivalent to those in \cref{fig:nSeWwSs_results_fig2} and \cref{fig:nSeWWwSs_results_fig3}. For each subfigure, we additionally present the verified ratio, $\left.\frac{n_{\mathrm{total}}}{n_{\mathrm{surf\ total}}}\right|_{d=6}$, of total vine code to surface code qubits at circuit-distance 6, and the predicted ratio, $\left.\frac{n_{\mathrm{total\ (fit)}}}{n_{\mathrm{surf\ total}}}\right|_{d=10}$, at $d_\mathrm{circuit}=10$ obtained from the quadratic fit.}
    \label{fig:distance_scaling_appendix}
\end{figure*}

In \cref{fig:distance_scaling_appendix},  we present the qubit overhead scaling of all codes presented in \cref{tab:performing_boundary_cuts} not already shown in \cref{sec:qubit_overhead_scaling} of the main text. 
For each subfigure in \cref{fig:distance_scaling_appendix}, we create patches of varying circuit distance, using the vine sequence and patch geometry as indicated. The circuit distances $d_{\mathrm{circuit}}$ indicated on the $x$-axis give the distances for both bases. 
These are obtained using the `connected cluster' algorithm in the dist-m4ri library~\cite{pryadko2024distm4ri} (see Ref.~\onlinecite{webster2026distancefindingalgorithmsquantumcodes} for a comparison of exact distance-finding algorithms and implementations).
For each patch, we find the number of data qubits $n$ and total number of qubits including data, measure and routing $n_{\mathrm{total}}$. We fit $n$ and $n_{\mathrm{total}}$ as $d_{\mathrm{circuit}}$ is varied to the ansatz as indicated, giving fit parameters for $A$, $B$, $C$ and $D$ in each subfigure. In each case, we also plot the equivalent number of qubits ($n$ data qubits and $n_{\mathrm{total}}$ data and measure qubits) needed to create $k$ surface code patches of equivalent circuit distance. Here, $k$ is the number of logical qubits encoded by the vine code as indicated. 

If the quadratic fit for $n_{\mathrm{total}}$, $C$, is smaller than $2k$, we expect that the efficiency of the vine code over the surface code to improve for large distances. In each subfigure, we give the verified ratio $\left.\frac{n_{\mathrm{total}}}{n_{\mathrm{surf\ total}}}\right|_{d=6}$ between $n_{\mathrm{total}}$ and the number of surface code qubits required for an equivalent encoding, $n_{\mathrm{surf\ total}}$, and $d_\mathrm{circuit}=6$, and we also provide the predicted ratio, $\left.\frac{n_{\mathrm{total\ (fit)}}}{n_{\mathrm{surf\ total}}}\right|_{d=10}$, at $d_\mathrm{circuit}=10$, obtained from the quadratic fit. The maximum asymptotic value of $\frac{n_{\mathrm{total}}}{n_{\mathrm{surf\ total}}}$ is given by $C/2k$.

\section{Vine code $\Delta$ sets}\label{app:delta_sets}

We define the set $\Delta_\text{odd}$, similarly to Ref.~\onlinecite{Geher2025Directional}. Firstly, define $\overrightarrow{Q_a^i Q_a^j}$ as the displacement vector pointing from qubit $Q_a^i$ to qubit $Q_a^j$, and let $\Delta$ be the multiset (counting multiplicity) of all such displacement vectors $\overrightarrow{Q_a^i Q_a^j}$ for $i<j$ and for any arbitrary measure qubit $a$. Finally, $\Delta_\text{odd}$ is defined as the set of all vectors in $\Delta$ that appear therein an odd number of times, along with the negatives of all such vectors. Two measure qubits separated by a combination of vectors from $\Delta_\text{odd}$ risk violating the stabiliser commutation requirements, at the circuit level. Hence, we require that any two measure qubits in a vine code that are separated by a vector from $\Delta_\text{odd}$ are of the same Pauli type (see Theorem 1 of Ref.~\onlinecite{Geher2025Directional}).

Additional to the two sets above we define $\Delta'$ as the multiset of displacement vectors and their negatives. In other words, it is the multiset of all $\overrightarrow{Q_a^iQ_a^j}$ for $i,j = 1,\ldots, w$. Then let $\Delta'_\text{odd}$ be the set of all vectors appearing in $\Delta'$ an odd number of times. An important fact for the construction of flip-vine codes is that any two stabilisers overlap an odd number of times if and only if the separation of their measure qubits, $\overrightarrow{ab}\in \Delta'_\text{odd}$. We prove this below.

\begin{prop}
    Consider a vine code with a length-$w$ step sequence. For two measure qubits, $a$ and $b$, define their stabiliser data qubit supports as sets $S_a$, $S_b$. Assume that $0\notin \Delta'$ and that $|S_a|=w$ for all measure qubits $a$. Then $\overrightarrow{ab}\in \Delta'_\text{odd}$ if and only if $|S_a\cap S_b|$ is odd. 
\end{prop}
\begin{proof}
    The set of qubits on which the stabilisers overlap is:
    \begin{align}
        S_a\cap S_b &= \lbrace (i,j)\, |\, \overrightarrow{Q_a^iQ_b^j} = 0\rbrace \\
        &= \lbrace (i,j) \, | \, \overrightarrow{Q_a^iQ_a^j} =- \overrightarrow{ab}\rbrace.
    \end{align}
    But recall that every $\overrightarrow{Q_a^iQ_a^j} \in \Delta'$. Hence, every time a vector $-\overrightarrow{ab}$ appears in $\Delta'$, there is an associated point in the overlap set, corresponding to $(i,j)$. And so $-\overrightarrow{ab}\in \Delta'_\text{odd}\iff \overrightarrow{ab}\in\Delta'_\text{odd}\iff |S_a\cap S_b|$ is odd.
\end{proof}

\section{Corners}
\label{app:corners}
Here, we describe the rules determining how stabilisers of $X$- and $Z$- type are truncated or removed around corners. 

As described in \cref{sec:pauli_bdrys}, Pauli boundaries can be defined by the parellelogram spanned by the vectors $v_1$ and $v_2$. $X$-stabilisers that have data qubit support outside the $Z$-boundary are removed, while $Z$ stabilisers with data qubits in this region are truncated. Correspondingly, $Z$-stabilisers that have data qubit support outside the $X$-boundary are removed, while $X$-stabilisers are truncated.

\cref{fig:corners} illustrates the Pauli boundary types for two vectors $v_1$ and $v_2$ and how the Pauli boundaries are transformed at each corner. At each corner, either $v$ or its additive inverse $-v$ extends from the corner away from the parallelogram, where $v \in \{ v_1, v_2 \}$. We use the angle bisector of these two outward extending rays to define the domain wall where the Pauli boundaries switch type. Following the rule that $X$-type ($Z$-type) stabilisers are removed if they have support in the $Z$ ($X$) boundary domain, we generate a protocol for removing/truncating each stabiliser type near the corners.  The figure highlights two stabiliser supports for the $\overrightarrow{nSeWwSs}$ shown as a vine. The two stabilisers shown have opposite Pauli type, shown by the red highlighed measure qubit stabilisers that are filled in with blue ($Z$-type measure qubit) or pink ($X$-type measure qubit). The $X$-type stabiliser on the lower right is removed because it includes one data qubit that passes into the blue boundary.  However, the $Z$-type stabiliser shown in the top left is truncated to the remaining data support on the parellelogram, as this stabiliser has no support in the pink, $X$ boundary region.

\begin{figure}
    \centering
    \includegraphics[width=\linewidth]{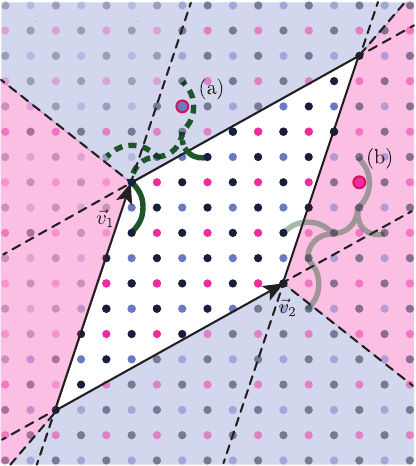}
    \caption{The Pauli boundaries at the corners of the $\overrightarrow{nSeWwSs}$ code, with the patch enclosed by the vectors $v_1$ and $v_2$. At each corner, we can define three rays extending away from the parallelogram: one of $\{ \pm v_1\}$, one of $\{ \pm v_2\}$, and their angular bisector which defines the domain wall between the $X$ boundary (shaded pink region) and the $Z$ boundary (shaed blue region). \textbf{(a)} The vine support of a $Z$-type stabiliser near the corner. The measure qubit's initial location is shown by the blue circle with a red outline. The data qubit support of this stabiliser falls within the parallelogram or $Z$ boundary (shaded blue). Therefore its support is truncated to the parallelogram. \textbf{(b)} The  vine support of an $X$-type stabiliser near the corner, with the measure qubit's initial location shown by the pink circle with a red outline. Here, the lower-most data qubit falls within the blue $Z$ boundary region, therefore this stabiliser is removed. }
    \label{fig:corners}
\end{figure}

\section{Minimising routing qubits and SWAP gates}\label{app:swap_heuristics}

\begin{algorithm}[ht]\label{alg:Long_swaps}
\caption{Delete Long Paths of SWAPS}
\KwIn{Vine code syndrome extraction circuit, $C$}

\KwOut{Valid vine code syndrome extraction circuit, $C'$, with long paths of SWAPs removed}

\SetKwFunction{Segments}{SwapSegments}
\SetKwData{LocalSkips}{LocalSkips}
\SetKwData{Skips}{Skips}
\SetKwData{CanSkip}{QubitCandidateSkips}
\SetKwData{Direction}{Direction}
\SetKwData{Count}{count}
\SetKwData{SwtoSkip}{SwapsToSkip}
\SetKwFunction{Top}{top}

$C' \gets C$\;
\ForEach{data/measure qubit $q$ in $C'$}{
    $\mathcal{S} \gets \Segments{q,\,C'}$ \tcp*{maximal runs of SWAPs acting on $q$}
    $\Skips \gets \emptyset$\;
    \ForEach{segment $S \in \mathcal{S}$}{
        let $M=(m_1,\dots,m_k)$ be the moves of $q$ induced by the SWAPs of $S$, with each $m_i \in \{(0,\pm 1),\,(\pm 1,0)\}$\;
        $\LocalSkips \gets \emptyset$\;
            \ForEach{$\Direction \in \lbrace (0,1),(1,0)\rbrace$}{
            $\Count \gets \min(\#\Direction\in M,\#(-\Direction)\in M)$ \tcp*{Number of pairs of forward-backward moves in $M$ in the direction of $\Direction$}
                \ForEach{$c =1,\ldots, \Count$}{
                    \If{$m_c=\Direction$}{
                        add $m_c$ to $\LocalSkips$\;
                    }
                }
                \ForEach{$c =1,\ldots, \Count$}{
                    \If{$m_c=-\Direction$}{
                        add $m_c$ to $\LocalSkips$ \tcp*{Try to delete equal number of forward and backward moves in the direction of $\Direction$}
                    }
                }
            }
        add $\LocalSkips$ to $\Skips$
    }
    \ForEach{candidate sequence of swaps to skip $\mathcal{C}\in\Skips$}{
    build trial circuit, $C_\text{trial}$, equivalent to $C'$ but with all SWAPs in $\mathcal{C}$ removed\;
    define the set of all non-swap gate locations ((target qubit location, gate layer) or (control qubit location, target qubit location, gate layer)) for $C'$, $\mathcal{N}$, and $C_\text{trial}$, $\mathcal{N}_\text{trial}$\;
    \If{$\mathcal{N}=\mathcal{N}_\text{trial}$}{
    $C'\gets C_\text{trial}$\;
    }
}
}
\Return{$C'$}
\end{algorithm}

\begin{algorithm}\label{alg:swap_meas_reset}
\caption{Delete SWAP-then-measure and reset-then-SWAP events}
\KwIn{Vine code syndrome extraction circuit, $C$}

\KwOut{Valid vine code syndrome extraction circuit, $C'$, with SWAP-then-measure and reset-then-SWAP events deleted}

$C'\gets C$\;
$C_\text{prev}\gets \emptyset$\;

\While{$C'\neq C_\text{prev}$}{
$C_\text{prev} \gets C'$\;
\ForEach{measure qubit $q$ in $C'$}{
\ForEach{SWAP gate, $S$, targeting $q$}{
\If{If there is no gate between $S$ and a measurement/reset of $q$}{
    $C'\gets (C'$ with $S$ deleted and the positions of reset/measurement of $q$ updated)}
}
}
}
\Return{$C'$}
\end{algorithm}

\begin{figure}
    \centering
    \includegraphics[width=\linewidth]{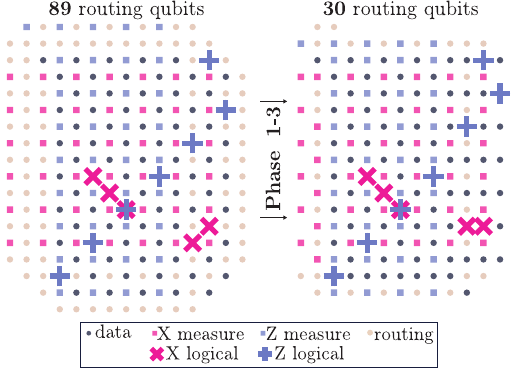}
    \caption{An example of the SWAP reduction techniques applied to a specific, $d=5$ patch of $\overrightarrow{nSENESs}$ code. After these are applied, $59$ redundant routing qubits are deleted. Stim circuits for both code patches can be found at Ref.~\onlinecite{nixon_2026_20734084}.}
    \label{fig:active_routing}
\end{figure}

We use several techniques for constructing circuits for vine codes with boundaries in such a way as to minimise the number of routing qubits and non-entangling SWAP gates required. We begin by constructing the circuit in the way described in the main text. The techniques that reduce the number of SWAPs can then be summarised as follows into three phases:

\textbf{Phase 1: Find valid shorter paths}

If a data/measure qubit undergoes a sequence of SWAPs that moves it between two locations, and a shorter path exists between the same two points, consider moving the data/measure qubit along the shorter path. These shorter paths are added to the circuit in a greedy fashion so long as the resulting circuit is equivalent to the original. We define two circuits as equivalent here if both have identical space-time locations of all non-swap gates (i.e., single-qubit, CX and CXSWAP gates). If this is the case, the circuits are performing the same syndrome extraction procedure and the underlying connectivity is the same, but they differ only in terms of SWAP gates between data/measure qubits and routing qubits. We summarise the procedure in Algorithm~\ref{alg:Long_swaps}.

\textbf{Phase 2: Remove SWAP-then-measure and reset-then-SWAP events}

After Phase 1, there may still be redundant SWAP sequences that span a measure/reset layer. If the SWAPs act on a data qubit, these will be removed already in Phase 1. But if they act on a measure qubit, the SWAPs will not form a contiguous, uninterrupted segment, and will not be removed in Phase 1. Since the syndrome measurement circuits reverse each QEC round, finding these removable SWAPs is equivalent to finding regions of the circuit in which a measure qubit is reset and then SWAPped with a routing qubit, or when a measure qubit is SWAPped with a routing qubit and then measured. In such cases, we can delete the SWAP gate and update the location of the measurement/reset. We must keep iterating through the circuit to look for these events, as chains of SWAPs can exist that will move resets/measurements multiple steps. We summarise the procedure in Algorithm~\ref{alg:swap_meas_reset}.

\textbf{Phase 3: Delete unused routing qubits}

Finally, after removing all SWAPs in the above manner, many routing qubits that were previously necessary may now be redundant. Hence, if any routing qubit is now not acted on by a gate in the circuit, we remove this routing qubit. An example of a vine code patch before and after these three phases are performed is shown in \cref{fig:active_routing}. Note the new positions of data and measure qubits in the updated patch.

\section{Generalised boundary details}\label{app:coloured_bdrys}

\begin{figure*}
    \centering
    \includegraphics[width=0.99\linewidth]{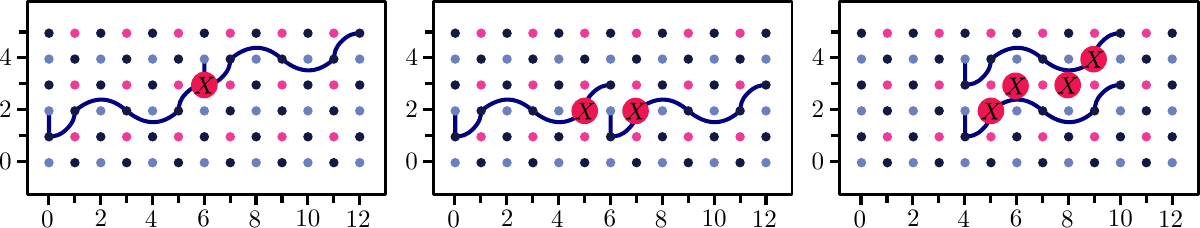}
    \caption{The three minimal-weight tunnelling operator motifs for the $\overrightarrow{sEEEn}$ code. The three Pauli-$X$ operators shown flip, in each case, only the two $Z$ stabilisers highlighted in blue.\label{fig:tun_motifs}}
\end{figure*}

The $\overrightarrow{sEEEn}$ (equivalently $\overrightarrow{NEEEN}$) vine code acts as a useful example to showcase the construction of generalised (or coloured) boundaries in vine codes, as it shares identical topological properties with the 2D color code. This is because it encodes $k=2$ logical qubits on the plane with Pauli boundary conditions.

We briefly recap the anyon content of the color code~\cite{Kesselring_2018_bdrys_twists,Kesselring_2024_anyon_condensation}. Like the color code, the $\overrightarrow{sEEEn}$ code has 9 non-trivial bosons which can be labelled by one of three colours $r$, $g$, $b$, and one of three Pauli labels $x$, $y$, $z$. For our purposes, we will treat $y$ anyons as the combination of $x$ and $z$ anyons. The bosons of the model can be summarised:
\begin{align}
\begin{array}{c|c|c}
r_x & g_x & b_x \\\hline
r_y & g_y & b_y \\\hline
r_z & g_z & b_z
\end{array}
\end{align}
We can identify these with the bosons of two copies of toric code~\cite{Kesselring_2018_bdrys_twists}: $r_x \mapsto e_1$, $r_z \mapsto m_2$, $b_x\mapsto e_2$, $b_z\mapsto m_1$, and these generate all other bosons via fusion.
The Lagrangian subgroups correspond to the rows and columns of the boson table, as each specifies three bosons that braid trivially. By measuring the short tunnelling operators of bosons in a row (column), we condense the complete Lagrangian subgroup of a Pauli (colour) type into the trivial phase. The remaining bosons are confined at this interface. This enables us to define six different boundary types for the $\overrightarrow{sEEEn}$ code. We have introduced two of these in Section~\ref{sec:pauli_bdrys} (corresponding to the first and third rows of the table), but we also define three more here, corresponding to the columns of the table. 

As explained in Section~\ref{sec:gen_bdrys}, to construct a coloured boundary of colour $c$, we measure out the tunnelling operators of $c_x$ and $c_z$ anyons. For the $\overrightarrow{sEEEn}$ code, there are three tunnelling operator motifs, which are shown in Fig.~\ref{fig:tun_motifs}. In the figure, we show the tunnelling operators for $b_z$ anyons and highlight the blue $Z$ stabilisers that are flipped by the $X$-type tunnelling operators. Horizontal translations of these motifs, and vertical translations by $[0,\pm 2]$, result in tunnelling operators for $b_z$, $g_z$ or $r_z$ anyons. Translations by $[1,1]$, combined with a Hadamard gate, create tunnelling operators for $r_x$ anyons. We once again obtain tunnelling operators for $g_x$ and $b_x$ via horizontal translations and translations by $[0,\pm 2]$. 

We construct a small patch of $\overrightarrow{sEEEn}$ code with coloured boundaries by measuring out tunnelling operators for $b$ anyons to the left and right of the patch, and then measure out $r$ anyon tunnelling operators to the top and bottom of the patch. This creates the stabiliser pattern shown in Fig.~\ref{fig:cld_bdry_ckts}, upon truncating the stabilisers that commute with all the measured tunnelling operators to lie within the patch. Note that this procedure can result in stabilisers that ``bleed" out of the patch and overlap with tunnelling operators. We can circumvent this issue by subsequently measuring out the qubits in the trivial region again in the $X$ or $Z$ basis. We then have some freedom to choose a complete set of stabiliser generators near the boundary to measure (e.g., we could measure $s_1$ and $s_2$, or $s_1$ and $s_1s_2$). The specific boundary stabilisers shown in Fig.~\ref{fig:cld_bdry_ckts} were chosen to minimise the number of new shapes generated (as these new shapes require changes to the syndrome measurement circuit; see below). We highlight the truncated stabilisers in the figure, and colour the bulk stabilisers in muted tones. We have shifted the starting positions of some of the data and measure qubits in this patch to allow for the measurement of all the truncated stabilisers without increasing the depth of the syndrome extraction circuit. In particular, the $X$ measure qubits on rows $0$ and $2$ have been moved in the direction of vector $[1,-1]$, and the data qubits directly above these, at $(2,3)$ and $(2,1)$ have been moved along vector $[-1, 1]$. 

\begin{figure*}
    \centering
    \includegraphics[width=0.99\linewidth]{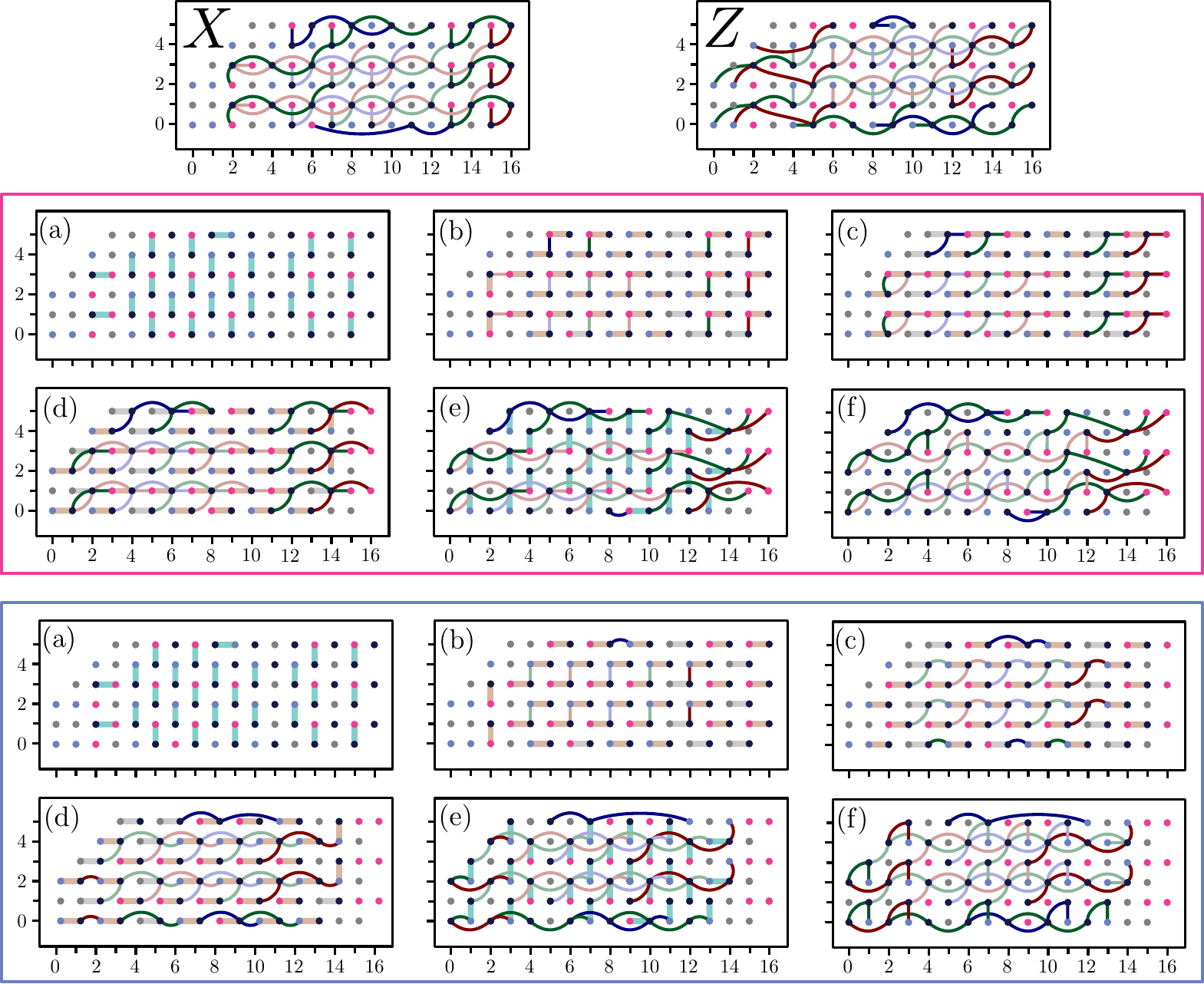}
    \caption{Syndrome extraction circuits for the $\overrightarrow{sEEEn}$ code with vertical blue and horizontal red boundaries (see Fig.~\ref{fig:coloured_bdry_stabs} for the boundaries and stabilisers of each colour). \textbf{(Top)} We show the complete $X$ (left) and $Z$ (right) stabilisers for this patch, with all data/measure qubits in their starting positions. \textbf{(Pink box; (a)-(f))} We show each step of the circuit and track how $X$ stabilisers are constructed. Teal (brown) lines indicate CX (CXSWAP) gates between data and measure qubits. Grey lines are SWAP gates. \textbf{(Blue box; (a)-(f))} The same steps of the circuit are shown and the evolution of the $Z$ stabilisers is shown.\label{fig:cld_bdry_ckts}}
\end{figure*}

Owing to some of the boundary stabilisers having a different shape to the bulk stabilisers, we need to change the syndrome measurement circuit close to the boundaries. In particular, note the different shapes of the four-body red $Z$ (green $X$) boundary stabilisers on the right (left) side of the patch, and the two-body blue stabilisers on the top and bottom boundaries of the patch. These require alterations to the direction of gates in the syndrome extraction circuit, while the remaining boundary stabilisers are simple truncations of the shapes of bulk stabilisers. In Fig.~\ref{fig:cld_bdry_ckts}, we provide all the gates for the syndrome measurement circuit and show how the $X$ and $Z$ stabilisers are formed over the course of the circuit. We choose locations for the measure qubits for the different-shaped stabilisers such that routing qubits would otherwise occupy these positions.

To see that the circuits do not violate the commutativity condition, we merely need to check that the different-shaped stabilisers are measured in a compatible fashion with their neighbours (all regular-shaped stabilisers or regular truncations are trivial, owing to the validity of the step sequence). We can focus on the right and bottom boundaries, as the top and left boundaries are the same for the sequence traced out in reverse. One can check that these stabilisers are compatible with all neighbours of the opposite Pauli type with which they overlap.

We provide a further example of the $\overrightarrow{sENWSw}$ flip-vine code, which is also equivalent to the 2D color code. We show this code on a triangular patch with red, green and blue boundaries. The stabilisers are supported on plaquettes that host $X$ and $Z$ type stabilisers. These are shown in Fig.~\ref{fig:flip_vine_ckt}. The circuit construction is far simpler for this example, as all boundary stabilisers are simple truncations of their bulk counterparts. In Fig.~\ref{fig:flip_vine_ckt}, we trace the evolution of one of the stabilisers supported on a blue plaquette. 

\begin{figure*}
    \centering
\begin{tikzpicture}
    \node[anchor=south west, inner sep=0] (img) at (0,0)
    {\includegraphics[width=0.99\linewidth]{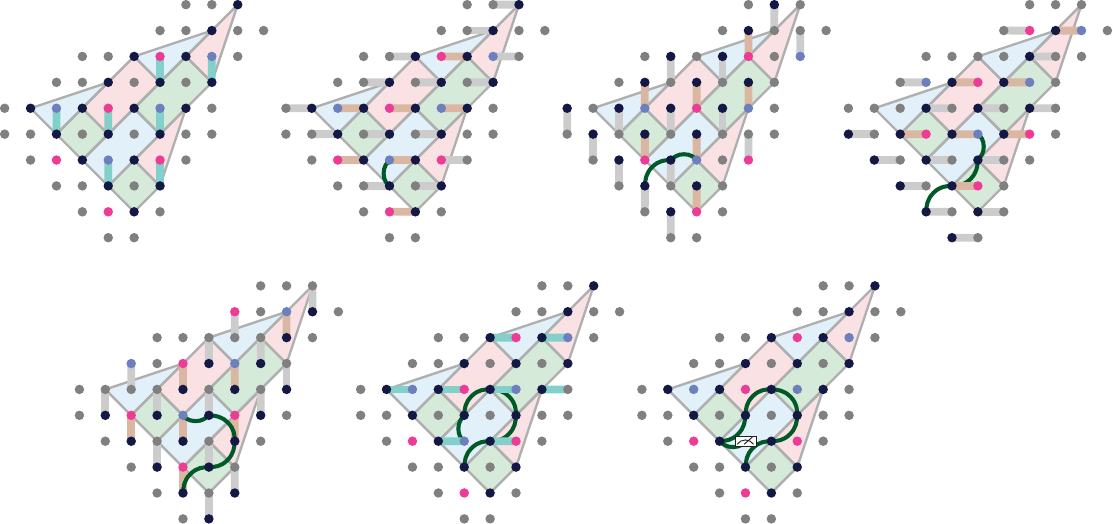}};
    \node[anchor=north west] at (0.65, 8.5) {(a)};
    \node[anchor=north west] at (5.0, 8.5) {(b)};
    \node[anchor=north west] at (9.35, 8.5) {(c)};
    \node[anchor=north west] at (13.7, 8.5) {(d)};
    \node[anchor=north west] at (1.75, 4.0) {(d)};
    \node[anchor=north west] at (6.3, 4.0) {(e)};
    \node[anchor=north west] at (10.75, 4.0) {(f)};
\end{tikzpicture}
    \caption{Syndrome measurement circuit for the $\overrightarrow{sENWSw}$ flip-vine code on a triangular patch. One $Z$ stabiliser is traced out through the steps of the circuit (a)-(e). The same measure qubit will measure an $X$ stabiliser with the same support in the next QEC round. The gate and qubit colour conventions are the same as in \cref{fig:vine_codes}, \cref{fig:boundary_stabilisers}, and \cref{fig:cld_bdry_ckts}.\label{fig:flip_vine_ckt}}
\end{figure*}

\section{Numerical Simulation Details}\label{app:numerics_details}

For our simulations, we use the Stim~\cite{Stim2021} package to build syndrome extraction circuits for both $X$ and $Z$ memory experiments (all stim circuits used for simulations can be found at Ref.~\onlinecite{nixon_2026_20734084}). For $X$ ($Z$) memory experiments, we start by resetting all data qubits in the $X$ ($Z$) basis and all measure and routing qubits in the $Z$ basis. In each QEC round, syndrome extraction is performed, and stabiliser measurements are compared with the last measurement record for that stabiliser to form ``detectors" (parity bits equivalent to the sum mod 2 of the two measurement outcomes). Detectors are $+1$ in the absence of errors and so are used to determine when errors have occurred. Circuits are run for $\max (5,d_\text{circuit})$ rounds of syndrome measurement, where $d_\text{circuit}$ is the minimum circuit distance of the code patch (either $3$, $5$, or $7$). At the end of the experiment, data qubits are read out in the $X$ ($Z$) basis for $X$ ($Z$) memory experiments.

Syndrome extraction circuits are compiled into a native gate set of: measurements in the $Z$ basis, resets in the $Z$ basis, Hadamard gates, CXSWAP, CZSWAP and CX gates (apart from in $X$ memory experiments wherein initial/final data qubit reset and readout is performed in the $X$ basis). The latter three gates are equivalent to gates native in superconducting devices (iSWAP and CZ) up to single-qubit Clifford rotations. Routing qubits are reset and measured in the $Z$ basis at the beginning and end of each QEC round and CZSWAP gates are used to route data/measure qubits via these routing qubits. We include a flag detector in the circuit for each routing qubit measurement. 

We assume a superconducting-inspired ``SI1000" noise model with noise strength $p$. This applies depolarising noise channels to qubits after each gate layer. Two-qubit gates are followed by a correlated two-qubit depolarising channel with strength $p$ and Hadamard gates by a single-qubit channel with strength $p/10$. Resets in the $Z$ ($X$) basis are followed by an $X$ ($Z$) error channel of strength $2p$. Measurements separately suffer a classical readout error with probability $5p$ and are followed by a depolarising channel of strength $p$. When a qubit idles during a unitary gate layer (measurement/reset layer), it suffers a depolarising channel of strength $p/10$ ($2p$).

We take $\geq 100,000$ Monte Carlo shots for each quantum memory circuit (or 1000 errors), collected using the Sinter tool in the Stim library. We use belief propagation with ordered statistics decoding (BPOSD) to decode the data~\cite{Panteleev_2021_BPOSD,Roffe_LDPC_Python_tools_2022,roffe_decoding_2020}. We only include detectors of the same Pauli type as the basis of the memory experiment in the decoding problem. This improved speeds of decoding considerably without noticeably affecting accuracy. For BPOSD, we use the same parameters as used in Ref.~\onlinecite{Geher2025Directional}. We use the min-sum method of BP and a parallel schedule, and use OSD-CS. The number of BP iterations and the OSD order are both set to 15, and the min-sum scaling factor is set to 0.625.

\bibliography{refs.bib}

\end{document}